\documentclass[%
 reprint,
 amsmath,amssymb,
 aps,
]{revtex4-1}
\usepackage{amsfonts}
\usepackage{amsmath}
\usepackage{booktabs}
\usepackage{amsthm}
\usepackage{graphicx}
\usepackage{dcolumn}
\usepackage{bm}

\begin{document}

\preprint{APS/123-QED}

\title{Novel methods to construct nonlocal sets of orthogonal product states in arbitrary bipartite high-dimensional system}
\author{Guang-Bao Xu$^{1,2}$}
\author{Dong-Huan Jiang$^{1}$}
\email{donghuan\_jiang@163.com}

\affiliation{
$^{1}$College of Mathematics and Systems Science, Shandong University of Science and Technology, Qingdao, 266590, China\\
$^{2}$State Key Laboratory of Networking and Switching Technology (Beijing University of Posts and Telecommunications), Beijing, 100876, China
}

\date{\today}

\begin{abstract}
Nonlocal sets of orthogonal product states (OPSs) are widely used in quantum protocols owing to their good property. In arXiv: 2003.03898, the authors consturcted some unextendible product bases in $\mathbb{C}^{m} \otimes \mathbb{C}^{n}$ quantum system for $n\geq m\geq3$. We find that a subset of their unextendible product basis (UPB) cannot be perfectly distinguished by local operations and classical communication (LOCC). We give a proof for the nonlocality of the subset with Vandermonde determinant and Kramer's rule. Meanwhile, we give a novel method to construct a nonlocal set with only $2(m+n)-4$ OPSs in $\mathbb{C}^{m} \otimes \mathbb{C}^{n}$ quantum system for $m\geq3$ and $n\geq3$. By comparing the number of OPSs in our nonlocal set with that of the existing results, we know that $2(m+n)-4$ is the minimum number of OPSs to construct a nonlocal and completable set in $\mathbb{C}^{m} \otimes \mathbb{C}^{n}$ quantum system so far. This means that we give the minimum number of a completable set of OPSs that cannot be perfectly distinguished by LOCC in arbitrary given space. Furthermore, we propose the concept of isomorphism of two nonlocal sets of OPSs. We analyze the relationship between different nonlocal sets in a given space using the new concept. Our work is of great help to understand the structure and classification of locally indistinguishable OPSs in arbitrary bipartite high-dimensional system.
\begin{description}
\item[PACS numbers]
03.65.Ud, 03.67.Mn
\end{description}
\end{abstract}

\pacs{Valid PACS appear here}
\maketitle


\section{\label{sec:level1}Introduction\protect}
Quantum nonlocality without entanglement (QNWE) is an important problem in quantum information theory. Many related works \cite{Bennett1999,Walgate2002,Fei2006,Halder2018,Halder2019,
Rout2019,Li2019,Band2018,Band2016, Feng2009,Sixia2015,CHB1999,Walgate2000,Niset2006,
Jiang2010,Yu2011} are proposed so far since a set of locally indistinguishable orthogonal product states (OPSs) can be used to design quantum protocols, such as quantum voting \cite{DHJiang2020} and quantum cryptography \cite{JWang2017,Rahaman2015,Guo2001}. Although great progress \cite{SHalder2019,Croke2017,XZhang2016,Duan2010} has been made in the field of QNWE, there are still some problems that have not been solved effectively. For example, the minimum number of OPSs to construct a nonlocal set in a given space and the classifications of different nonlocal sets of OPSs.

As we know, a set of OPSs can be exactly discriminated by global positive-operator-valued measurements. However, this task may not be accomplished if only local operations and classical communication (LOCC) are permitted. Many people intuitively thought that some sets of quantum states cannot be reliably discriminated by LOCC because of quantum entanglement. Bennett \emph{et al.} \cite{Bennett1999} firstly showed that a set of 9 OPSs cannot be reliably discriminated by LOCC in  $\mathbb{C}^{3} \otimes \mathbb{C}^{3}$. This counterintuitive phenomenon is called QNWE by Bennett \emph{et al.}. Encouraged by Bennett \emph{et al.}'s work, many people began to engage in the research of QNWE and a lot of results \cite{Walgate2002,Zhang2014} were proposed.

Now we introduce the development of the construction methods of locally indistinguishable OPSs in bipartite quantum systems. Zhang \emph{et al.} \cite{Zhang2014} gave a method to construct a complete orthogonal product basis (OPB) that cannot be exactly distinguished by LOCC in $\mathbb{C}^{d} \otimes \mathbb{C}^{d}$ quantum system for $d\geq3$. Based on this result, Wang \emph{et al.} \cite{Wang2015} pointed out that a subset of Zhang \emph{et al.}'s OPB is still nonlocal. Meanwhile, they generalized their method to a more general $\mathbb{C}^{m} \otimes \mathbb{C}^{n}$ quantum system for $m\geq3$ and $n\geq3$. Zhang \emph{et al.} \cite{Zhang2015} constructed a nonlocal set with $4d-4$ OPSs in $\mathbb{C}^{d} \otimes \mathbb{C}^{d}$ quantum system for $d\geq 3$ and then generalized this result to a more general case \cite{Zhang2016} in $\mathbb{C}^{m} \otimes \mathbb{C}^{n}$ quantum system for $3\leq m\leq n$. Recently, Zhang \emph{et al.} \cite{Zhangxiaoqian2017} presented a nonlocal set of $3(n+m)-8$ OPSs in $\mathbb{C}^{m}\otimes\mathbb{C}^{n}$ quantum system. Although many achievements have been made in constructing nonlocal sets of OPSs, most of the results contain a lot of elements and have complex structures. It is interesting to find the minimum number of elements to form a completable set of OPSs which cannot be perfectly distinguished by LOCC in a general $\mathbb{C}^{m} \otimes \mathbb{C}^{n}$ quantum system for $m\geq 3$ and $n\geq 3$.

In Ref. \cite{SHalder2019}, Halder \emph{et al.} constructed an unextendible product basis (UPB) in $\mathbb{C}^{d}\otimes\mathbb{C}^{d}$ quantum system with $d$ is odd and $d\geq 5$. Inspired by their work, Shi \emph{et al.} give a more general method to construct UPBs in $\mathbb{C}^{m}\otimes\mathbb{C}^{n}$ quantum system with $3\leq m \leq n$ in Ref. \cite{Fei2020}. As we know, if a set of quantum states cannot be perfectly distinguished by LOCC, its subset may not necessarily be locally indistinguishable by LOCC. It is interesting to find the minimum number of OPSs to form a completable set that cannot be perfectly distinguished in a given Hilbert space. Based on this idea, we find that a subset of the UPB constructed by Shi \emph{et al.}, which only has $2(m+n)-4$ elements, cannot be perfectly distinguished by LOCC in $\mathbb{C}^{m}\otimes\mathbb{C}^{n}$ for $3\leq m\leq n$. Most important of all, we give a novel method to construct a nonlocal set of OPSs in $\mathbb{C}^{m}\otimes\mathbb{C}^{n}$ quantum sysytem for $m\geq 3$ and $n\geq 3$ in this paper. Both the subset and the novel set have only $2(m+n)-4$ members, which is the minimum number of OPSs to form a completable set that cannot be perfectly distinguished by LOCC in a general $\mathbb{C}^{m}\otimes\mathbb{C}^{n}$ quantum system for $m\geq 3$ and $n\geq 3$ so far. On the other hand, to analyze the structures of bipartite nonlocal sets constructed by our method and Zhang \emph{et al.}'s, we give the concept of isomorphism of two nonlocal sets of OPSs. All the results improve the theory of QNWE. The rest of this paper is organized as follows. In sec. \uppercase\expandafter{\romannumeral2}
, some preliminaries that will be used in the following sections are introduced. In sec. \uppercase\expandafter{\romannumeral3}, we firstly give a subset of a UPB of Shi \emph{et al.}, then prove it cannot be perfectly distinguished by LOCC.
In sec. \uppercase\expandafter{\romannumeral4}, we firstly give our novel method to construct a nonlocal set of OPSs in $\mathbb{C}^{m}\otimes\mathbb{C}^{n}$ for $m\geq 3$ and $n\geq 3$. Then, we compare the number of elements of our nonlocal set with that of the existing ones. In sec. \uppercase\expandafter{\romannumeral5}, we propose the concept of isomorphism of two sets of bipartite OPSs. In sec. \uppercase\expandafter{\romannumeral6}, a brief conclusion is given.

\section{\label{sec:level1}Preliminaries}

\theoremstyle{remark}
\newtheorem{definition}{\indent Definition}
\newtheorem{lemma}{\indent Lemma}
\newtheorem{theorem}{\indent Theorem}
\newtheorem{corollary}{\indent Corollary}
\def\QEDclosed{\mbox{\rule[0pt]{1.3ex}{1.3ex}}}
\def\QED{\QEDclosed}
\def\proof{\indent{\em Proof}.}
\def\endproof{\hspace*{\fill}~\QED\par\endtrivlist\unskip}

Some preliminaries, which will be used in what follows, are given in this section.
\begin{definition} \cite{Xu2016,xu2016}
If a set of OPSs cannot be exactly discriminated by LOCC, we say it is locally indistinguishable or nonlocal.
\end{definition}
\begin{definition} \cite{SHalder2019,Walgate2000,Sixia2015,Zhang2016,Wang2015}
A quantum measurement is trivial if all its positive operator-valued measure elements are proportional to identity operator.
\end{definition}
\begin{lemma} (Vandermonde determinant \cite{Depart2014})
The following determinant is called Vandermonde determinant since it is firstly researched by Vandermonde.
\begin{equation}
\nonumber
\begin{split}
\left|
  \begin{array}{ccccc}
    1       &a_{1}   &a_{1}^{2} &\cdots &a_{1}^{n-1}\\
    1       &a_{2}   &a_{2}^{2} &\cdots &a_{2}^{n-1}\\
    1       &a_{3}   &a_{3}^{2} &\cdots &a_{3}^{n-1}\\
    \vdots  &\vdots  &\vdots    &\ddots &\vdots\\
    1       &a_{n}   &a_{n}^{2} &\cdots &a_{n}^{n-1}\\
  \end{array}
\right|
=\prod_{1\leq j< t\leq n}(a_{t}-a_{j}),
\end{split}
\end{equation}
where $t,\,j,\,n$ are positive integers.
\end{lemma}

\begin{lemma} (Kramer's rule \cite{Department2014}) A system of equations
\begin{equation}
\nonumber
\left\{
\begin{aligned}
\alpha_{11}x_{1}+\alpha_{12}x_{2}+\cdots+\alpha_{1n}x_{n}=\beta_{1}\\
\alpha_{21}x_{1}+\alpha_{22}x_{2}+\cdots+\alpha_{2n}x_{n}=\beta_{2}\\
\vdots\qquad \qquad \qquad\quad \\
\alpha_{n1}x_{1}+\alpha_{n2}x_{2}+\cdots+\alpha_{nn}x_{n}=\beta_{n}\\
\end{aligned}
\right.
\end{equation}
has a unique solution if its coefficient determinant
\begin{equation}
\nonumber
\begin{split}
\left|
  \begin{array}{cccc}
    \alpha_{11}   &\alpha_{12}    &\cdots   &\alpha_{1n}\\
    \alpha_{21}   &\alpha_{22}    &\cdots   &\alpha_{2n}\\
    \vdots        &\vdots         &\ddots   &\vdots\\
    \alpha_{n1}   &\alpha_{n2}    &\cdots   &\alpha_{nn}
  \end{array}
\right|
\neq 0,
\end{split}
\end{equation}
where $\alpha_{\lambda\mu}$ and $\beta_{\lambda}$ are plurals for $\lambda=1,$ $2,$ $\cdots,$ $n$ and $\mu=1$, $2,$ $\cdots,$ $n.$
\end{lemma}
\begin{lemma} \cite{Gai2007}
If $\omega=e^{\frac{2\pi \sqrt{-1}}{d}}$, we have
$$\omega^{t} \neq \omega^{j}$$
and $$(\omega^{p})^{d}=1$$ for $1\leq t<j\leq d$ and $1\leq p\leq d$, where $t$, $j$, $p$ and $d$ are integers; and $d\geq 2$.
\end{lemma}
\section{Nonlocal subset of Shi et al.'s UPB}
In \cite{Fei2020}, Shi \emph{et al.} gave a more general method to construct a UPB in $\mathbb{C}^{m}\otimes\mathbb{C}^{n}$ quantum system with $3\leq m \leq n$. As we know, if a set of quantum states cannot be perfectly distinguished by LOCC, its subset is not necessarily nonlocal. However, we found a subset of the UPB is still LOCC indistinguishable. The subset only has $2(m+n)-4$ members, which is the minimum number of OPSs to form a completable set that cannot be perfectly distinguished by LOCC in $\mathbb{C}^{m}\otimes\mathbb{C}^{n}$ for $m\geq 3$ and $n\geq 3$ so far.

To make readers understand the structure and the nonlocality of the subset of Shi \emph{et al.}'s UPB, we first give a special case, \emph{i.e.}, Theorem 1. It should be noted that all the product states in this paper are not normalized for convenience.

\begin{figure}
\small
\centering
\includegraphics[scale=0.6]{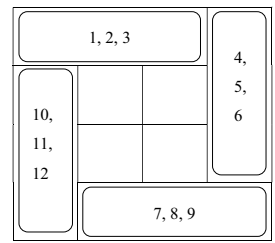}
\caption{Nonlocal subset of Shi \emph{et al.}'s UPB in $\mathbb{C}^{4} \otimes \mathbb{C}^{4}$ quantum system. Here $j$ denotes the product states $|\phi_{j}\rangle$ for $j=1,\,2,\,\cdots,\,12.$}
\end{figure}

\begin{theorem}
In $\mathbb{C}^{4} \otimes \mathbb{C}^{4}$ quantum system, the states in Eqs. (1) construct a nonlocal set of OPSs, \emph{i.e.}, these states cannot be perfectly distinguished by LOCC.
\begin{equation}
\begin{aligned}
|\phi_{1}\rangle=|0\rangle_{A}(|0\rangle+|1\rangle+|2\rangle)_{B},\\
|\phi_{2}\rangle=|0\rangle_{A}(|0\rangle+\omega|1\rangle+\omega^{2}|2\rangle)_{B},\\
|\phi_{3}\rangle=|0\rangle_{A}[|0\rangle+\omega^{2}|1\rangle+(\omega^{2})^{2}|2\rangle]_{B},\\ |\phi_{4}\rangle=(|0\rangle+|1\rangle+|2\rangle)_{A}|3\rangle_{B},\\
|\phi_{5}\rangle=(|0\rangle+\omega|1\rangle+\omega^{2}|2\rangle)_{A}|3\rangle_{B},\\
|\phi_{6}\rangle=[|0\rangle+\omega^{2}|1\rangle+(\omega^{2})^{2}|2\rangle]_{A}|3\rangle_{B},\\
|\phi_{7}\rangle=|3\rangle_{A}(|1\rangle+|2\rangle+|3\rangle)_{B},\\
|\phi_{8}\rangle=|3\rangle_{A}(|1\rangle+\omega|2\rangle+\omega^{2}|3\rangle)_{B},\\
|\phi_{9}\rangle=|3\rangle_{A}[|1\rangle+\omega^{2}|2\rangle+(\omega^{2})^{2}|3\rangle]_{B},\\
|\phi_{10}\rangle=(|1\rangle+|2\rangle+|3\rangle)_{A}|0\rangle_{B},\\
|\phi_{11}\rangle=(|1\rangle+\omega|2\rangle+\omega^{2}|3\rangle)_{A}|0\rangle_{B},\\
|\phi_{12}\rangle=[|1\rangle+\omega^{2}|2\rangle+(\omega^{2})^{2}|3\rangle]_{A}|0\rangle_{B},
\end{aligned}
\end{equation}
where $\omega=e^{\frac{2\pi \sqrt{-1}}{3}}$. FIG. 1 shows a clear and intuitive structure of these states.
\end{theorem}
\begin{proof}
It should be noted that the proof method is originally given in Refs. \cite{Walgate2002,
Sixia2015}. To discriminate one of the 12 states, one party has to start with a nontrivial measurement of preserving orthogonality, \emph{i.e.}, the post-measurement states should be mutually orthogonal and  not all $M_{k}^{\dag}M_{k}$ are proportional to identity operator.

Suppose that Alice firstly starts with a set of general $4\times4$ positive operator-valued measurement (POVM) elements $\{M_{k}^{\dag}M_{k}: k=1,2,\cdots,l\}$, where
  \begin{equation}
\begin{split}
M_{k}^{\dag}M_{k}=
\left[
  \begin{array}{cccc}
    a_{00}^{k}&a_{01}^{k} &a_{02}^{k} &a_{03}^{k}\\
    a_{10}^{k}&a_{11}^{k} &a_{12}^{k} &a_{13}^{k}\\
    a_{20}^{k}&a_{21}^{k} &a_{22}^{k} &a_{23}^{k}\\
    a_{30}^{k}&a_{31}^{k} &a_{32}^{k} &a_{33}^{k}\\
  \end{array}
\right]
\end{split}\nonumber
\end{equation}
in the basis $\{|0\rangle,\,|1\rangle,\,|2\rangle,\,|3\rangle\}$. The post-measurement states $\{(M_{k}\otimes I_{4\times 4})|\phi_{j}\rangle:$ $j=1,$ $2$, $\cdots$, $12\}$ should preserve their orthogonality.

Because $(M_{k}\otimes I_{4\times4})|\phi_{1}\rangle$ is orthogonal to $(M_{k}\otimes I_{4\times4})|\phi_{10}\rangle$, $(M_{k}\otimes I_{4\times4})|\phi_{11}\rangle$ and $(M_{k}\otimes I_{4\times4})|\phi_{12}\rangle$, \emph{i.e.},
\begin{equation}
\nonumber
\left\{
\begin{aligned}
\langle\phi_{1}|(M_{k}^{\dag}M_{k}\otimes I_{4\times 4})|\phi_{10}\rangle=0\\
\langle\phi_{1}|(M_{k}^{\dag}M_{k}\otimes I_{4\times 4})|\phi_{11}\rangle=0\\
\langle\phi_{1}|(M_{k}^{\dag}M_{k}\otimes I_{4\times 4})|\phi_{12}\rangle=0
\end{aligned}
\right.
\end{equation}
and
\begin{equation}
\nonumber
\left\{
\begin{aligned}
\langle\phi_{10}|(M_{k}^{\dag}M_{k}\otimes I_{4\times 4})|\phi_{1}\rangle=0\\
\langle\phi_{11}|(M_{k}^{\dag}M_{k}\otimes I_{4\times 4})|\phi_{1}\rangle=0\\
\langle\phi_{12}|(M_{k}^{\dag}M_{k}\otimes I_{4\times 4})|\phi_{1}\rangle=0
\end{aligned}
\right.,
\end{equation}
we get two systems of linear equations
\begin{equation}
\left\{
\begin{aligned}
a_{01}^{k}+a_{02}^{k}+a_{03}^{k}=0\qquad\quad\,\,\\
a_{01}^{k}+\omega a_{02}^{k}+\omega^{2}a_{03}^{k}=0\quad\,\,\,\,\\
a_{01}^{k}+\omega^{2}a_{02}^{k}+(\omega^{2})^{2}a_{03}^{k}=0
\end{aligned}
\right.
\end{equation}
and
\begin{equation}
\left\{
\begin{aligned}
a_{10}^{k}+a_{20}^{k}+a_{30}^{k}=0\qquad\quad\,\,\\
a_{10}^{k}+\overline{\omega} a_{20}^{k}+\overline{\omega}^{2}a_{30}^{k}=0\quad\,\,\,\,\\
a_{10}^{k}+\overline{\omega}^{2}a_{20}^{k}+(\overline{\omega}^{2})^{2}a_{30}^{k}=0
\end{aligned}
\right.,
\end{equation}
where $\overline{\omega}$ is the conjugate complex number of $\omega$.

By Lemma 1-3, we get the unique solution of Eqs. (2),
\begin{eqnarray}
\begin{aligned}
a_{01}^{k}=a_{02}^{k}=a_{03}^{k}=0\\
\end{aligned}
\end{eqnarray}
and the unique solution of Eqs. (3),
\begin{eqnarray}
\begin{aligned}
a_{10}^{k}=a_{20}^{k}=a_{30}^{k}=0.
\end{aligned}
\end{eqnarray}
Similarly, we can get two systems of linear equations
\begin{equation}
\left\{
\begin{aligned}
a_{30}^{k}+a_{31}^{k}+a_{32}^{k}=0\qquad\quad\,\,\\
a_{30}^{k}+\omega a_{31}^{k}+\omega^{2}a_{32}^{k}=0\quad\,\,\,\,\\
a_{30}^{k}+\omega^{2}a_{31}^{k}+(\omega^{2})^{2}a_{32}^{k}=0
\end{aligned}
\right.
\end{equation}
and
\begin{equation}
\left\{
\begin{aligned}
a_{03}^{k}+a_{13}^{k}+a_{23}^{k}=0\qquad\quad\,\,\\
a_{03}^{k}+\overline{\omega} a_{13}^{k}+\overline{\omega}^{2}a_{23}^{k}=0\quad\,\,\,\,\\
a_{03}^{k}+\overline{\omega}^{2}a_{13}^{k}+(\overline{\omega}^{2})^{2}a_{23}^{k}=0
\end{aligned}
\right.
\end{equation}
since $(M_{k}\otimes I_{4\times4})|\phi_{7}\rangle$ is orthogonal to $(M_{k}\otimes I_{4\times4})|\phi_{4}\rangle$, $(M_{k}\otimes I_{4\times4})|\phi_{5}\rangle$ and $(M_{k}\otimes I_{4\times4})|\phi_{6}\rangle$. By Lemma 1-3, we get the unique solution of Eqs. (6),
\begin{equation}
\begin{aligned}
a_{30}^{k}=a_{31}^{k}=a_{32}^{k}=0
\end{aligned}
\end{equation} and the unique solution of Eqs. (7),
\begin{equation}
\begin{aligned}
a_{03}^{k}=a_{13}^{k}=a_{23}^{k}=0.
\end{aligned}
\end{equation}

Because these three states $(M_{k}\otimes I_{4\times4})|\phi_{4}\rangle$, $(M_{k}\otimes I_{4\times4})|\phi_{5}\rangle$ and $(M_{k}\otimes I_{4\times4})|\phi_{6}\rangle$ are mutually orthogonal, we have
\begin{equation}
\nonumber
\left\{
\begin{aligned}
\langle\phi_{4}|(M_{k}^{\dag}M_{k}\otimes I_{4\times 4})|\phi_{5}\rangle=0\\
\langle\phi_{4}|(M_{k}^{\dag}M_{k}\otimes I_{4\times 4})|\phi_{6}\rangle=0
\end{aligned}
\right.,
\end{equation}
\begin{equation}
\nonumber
\left\{
\begin{aligned}
\langle\phi_{5}|(M_{k}^{\dag}M_{k}\otimes I_{4\times 4})|\phi_{4}\rangle=0\\
\langle\phi_{5}|(M_{k}^{\dag}M_{k}\otimes I_{4\times 4})|\phi_{6}\rangle=0
\end{aligned}
\right.
\end{equation}
and
\begin{equation}
\nonumber
\left\{
\begin{aligned}
\langle\phi_{6}|(M_{k}^{\dag}M_{k}\otimes I_{4\times 4})|\phi_{4}\rangle=0\\
\langle\phi_{6}|(M_{k}^{\dag}M_{k}\otimes I_{4\times 4})|\phi_{5}\rangle=0
\end{aligned}
\right..
\end{equation}
That is,
\begin{equation}
\left\{
\begin{aligned}
\sum_{p=1}^{2}(\omega^{p}\sum_{j=0}^{2}a_{jp}^{k})
=-\sum_{j=0}^{2}a_{j0}^{k}\quad\,\,\\
\sum_{p=1}^{2}[(\omega^{2})^{p}\sum_{j=0}^{2}a_{jp}^{k}]
=-\sum_{j=0}^{2}a_{j0}^{k}
\end{aligned}
\right.,
\end{equation}
\begin{equation}
\left\{
\begin{aligned}
\sum_{p=1}^{2}\sum_{j=0}^{2}\overline{\omega}^{j}a_{jp}^{k}
=-\sum_{j=0}^{2}\overline{\omega}^{j}a_{j0}^{k}\qquad\quad\\
\sum_{p=1}^{2}[(\omega^{2})^{p}\sum_{j=0}^{2}\overline{\omega}^{j}a_{jp}^{k}]
=-\sum_{j=0}^{2}\overline{\omega}^{j}a_{j0}^{k}\\
\end{aligned}
\right.,
\end{equation}
and
\begin{equation}
\left\{
\begin{aligned}
\sum_{p=1}^{2}\sum_{j=0}^{2}(\overline{\omega}^{2})^{j}a_{jp}^{k}
=-\sum_{j=0}^{2}(\overline{\omega}^{2})^{j}a_{j0}^{k}\quad\,\,\,\\
\sum_{p=1}^{2}[\omega^{p}\sum_{j=0}^{2}(\overline{\omega}^{2})^{j}a_{jp}^{k}]
=-\sum_{j=0}^{2}(\overline{\omega}^{2})^{j}a_{j0}^{k}\\
\end{aligned}
\right.,
\end{equation}
where $\overline{\omega}$ is the conjugate complex number of $\omega$.

For simplicity, we denote the coefficient determinants of Eqs. (10), (11) and (12) as $D_{1}$, $D_{2}$ and $D_{3}$, respectively. Since
\begin{equation}
\begin{split}
\nonumber
D_{1}=
\left|
  \begin{array}{cc}
   \omega        &\omega^{2}        \\
   \omega^{2}    &(\omega^{2})^{2}   \\
  \end{array}
\right|
=\omega^{3}\left|
  \begin{array}{cc}
   1    &\omega        \\
   1    &\omega^{2}   \\
  \end{array}
\right|
=\omega^{2}-\omega \neq 0,\\
\end{split}
\end{equation}
\begin{equation}
\nonumber
\begin{split}
D_{2}=
\left|
  \begin{array}{cc}
   1             &1        \\
   \omega^{2}    &(\omega^{2})^{2}   \\
  \end{array}
\right|
=(\omega^{2})^{2}-\omega^{2} \neq 0,
\end{split}\qquad\qquad
\end{equation}
\begin{equation}
\nonumber
\begin{split}
D_{3}=
\left|
  \begin{array}{cc}
   1        &1        \\
   \omega   &\omega^{2} \\
  \end{array}
\right|
=\omega^{2}-\omega \neq 0,
\end{split}\qquad\qquad\qquad\qquad
\end{equation}
we obtain the unique solutions of Eqs. (10), (11) and (12), respectively, \emph{i.e.},
\begin{equation}
\left\{
\begin{aligned}
\sum_{j=0}^{2}a_{j1}^{k}=\sum_{j=0}^{2}a_{j0}^{k}\\
\sum_{j=0}^{2}a_{j2}^{k}=\sum_{j=0}^{2}a_{j0}^{k}
\end{aligned}
\right.,
\end{equation}
\begin{equation}
\left\{
\begin{aligned}
\sum_{j=0}^{2}\overline{\omega}^{j}a_{j1}^{k}=
\overline{\omega}\sum_{j=0}^{2}\overline{\omega}^{j}a_{j0}^{k}\\
\sum_{j=0}^{2}\overline{\omega}^{j}a_{j2}^{k}=
\overline{\omega}^{2}\sum_{j=0}^{2}\overline{\omega}^{j}a_{j0}^{k}\\
\end{aligned}
\right.,\quad
\end{equation}
and
\begin{equation}
\left\{
\begin{aligned}
\sum_{j=0}^{2}(\overline{\omega}^{2})^{j}a_{j1}^{k}=\overline{\omega}^{2}
\sum_{j=0}^{2}(\overline{\omega}^{2})^{j}a_{j0}^{k}\quad\\
\sum_{j=0}^{2}(\overline{\omega}^{2})^{j}a_{j2}^{k}=(\overline{\omega}^{2})^{2}
\sum_{j=0}^{2}(\overline{\omega}^{2})^{j}a_{j0}^{k}\\
\end{aligned}
\right..
\end{equation}
By Eqs. (5), (13), (14) and (15), we get
\begin{equation}
\left\{
\begin{aligned}
\sum_{j=0}^{2}a_{j1}^{k}=a_{00}^{k}\qquad\quad\\
\sum_{j=0}^{2}\overline{\omega}^{j}a_{j1}^{k}=\overline{\omega} a_{00}^{k}\quad\,\,\\
\sum_{j=0}^{2}(\overline{\omega}^{2})^{j}a_{j1}^{k}=\overline{\omega}^{2}a_{00}^{k}
\end{aligned}
\right.
\end{equation}
and
\begin{equation}
\left\{
\begin{aligned}
\sum_{j=0}^{2}a_{j2}^{k}=a_{00}^{k}\qquad\qquad\\
\sum_{j=0}^{2}\overline{\omega}^{j}a_{j2}^{k}=\overline{\omega}^{2} a_{00}^{k}\qquad\\
\sum_{j=0}^{2}(\overline{\omega}^{2})^{j}a_{j2}^{k}=(\overline{\omega}^{2})^{2}a_{00}^{k}
\end{aligned}
\right..
\end{equation}
By Lemma 1-3, we easily get the unique solutions of Eqs. (16) and Eqs. (17), respectively, \emph{i.e.},
\begin{equation}
\left\{
\begin{aligned}
a_{11}^{k}=a_{00}^{k}\qquad\\
a_{01}^{k}=a_{21}^{k}=0
\end{aligned}
\right.
\end{equation}
and
\begin{equation}
\left\{
\begin{aligned}
a_{22}^{k}=a_{00}^{k}\quad\,\,\,\,\,\\
a_{02}^{k}=a_{12}^{k}=0
\end{aligned}
\right..
\end{equation}

Similarly, because the three product states $\{(M_{k}\otimes I_{4\times4})|\phi_{10}\rangle$, $(M_{k}\otimes I_{4\times4})|\phi_{11}\rangle$ and $(M_{k}\otimes I_{4\times4})|\phi_{12}\rangle\}$ are mutually orthogonal, we have
\begin{equation}
\left\{
\begin{aligned}
\sum_{p=0}^{1}[\omega^{p}\sum_{j=1}^{3}a_{j(p+1)}^{k}]
=-\omega^{2}\sum_{j=1}^{3}a_{j3}^{k}\qquad\,\,\\
\sum_{p=0}^{1}[(\omega^{2})^{p}\sum_{j=1}^{3}a_{j(p+1)}^{k}]
=-(\omega^{2})^{2}\sum_{j=1}^{3}a_{j3}^{k}\\
\end{aligned}
\right.,
\end{equation}
\begin{equation}
\left\{
\begin{aligned}
\sum_{p=0}^{1}\sum_{j=1}^{3}\overline{\omega}^{j-1}a_{j(p+1)}^{k}
=-\sum_{j=1}^{3}\overline{\omega}^{j-1}a_{j3}^{k}\qquad\qquad\quad\,\\
\sum_{p=0}^{1}[(\omega^{2})^{p}\sum_{j=1}^{3}\overline{\omega}^{j-1}a_{j(p+1)}^{k}]
=-(\omega^{2})^{2}\sum_{j=1}^{3}\overline{\omega}^{j-1}a_{j3}^{k}\\
\end{aligned}
\right.
\end{equation}
and
\begin{equation}
\left\{
\begin{aligned}
\sum_{p=0}^{1}\sum_{j=1}^{3}(\overline{\omega}^{2})^{j-1}a_{j(p+1)}^{k}
=-\sum_{j=1}^{3}(\overline{\omega}^{2})^{j-1}a_{j3}^{k}\qquad\,\,\,\,\\
\sum_{p=0}^{1}[\omega^{p}\sum_{j=1}^{3}(\overline{\omega}^{2})^{j-1}a_{j(p+1)}^{k}]
=-\omega^{2}\sum_{j=1}^{3}(\overline{\omega}^{2})^{j-1}a_{j3}^{k}\\
\end{aligned}
\right..
\end{equation}
For simplicity, we denote the coefficient determinants of Eqs. (20), (21) and (22) as $D_{4}$, $D_{5}$ and $D_{6}$, respectively. Since
\begin{equation}
\begin{split}
\nonumber
D_{4}=\left|
  \begin{array}{cc}
   1   &\omega        \\
   1   &\omega^{2}   \\
  \end{array}
\right| \neq 0,
\end{split}
\end{equation}
\begin{equation}
\nonumber
\begin{split}
D_{5}=\left|
  \begin{array}{cc}
   1        &1        \\
   1    &\omega^{2}   \\
  \end{array}
\right| \neq 0,
\end{split}
\end{equation}
\begin{equation}
\nonumber
\begin{split}
D_{6}=\left|
  \begin{array}{cc}
   1    &1       \\
   1    &\omega   \\
  \end{array}
\right|\neq 0,
\end{split}
\end{equation}
we can get the unique solutions of Eqs. (20), (21) and (22), respectively, \emph{i.e.},
\begin{equation}
\left\{
\begin{aligned}
\sum_{j=1}^{3}a_{j1}^{k}=\sum_{j=1}^{3}a_{j3}^{k}\\
\sum_{j=1}^{3}a_{j2}^{k}=\sum_{j=1}^{3}a_{j3}^{k}
\end{aligned}
\right.,
\end{equation}
\begin{equation}
\left\{
\begin{aligned}
\sum_{j=1}^{3}\overline{\omega}^{j-1}a_{j1}^{k}=
\omega^{2}\sum_{j=1}^{3}\overline{\omega}^{j-1}a_{j3}^{k}\\
\sum_{j=1}^{3}\overline{\omega}^{j-1}a_{j2}^{k}=
\omega\sum_{j=1}^{3}\overline{\omega}^{j-1}a_{j3}^{k}\,\,\\
\end{aligned}
\right.,
\end{equation}
and
\begin{equation}
\left\{
\begin{aligned}
\sum_{j=1}^{3}(\overline{\omega}^{2})^{j-1}a_{j1}^{k}=(\omega^{2})^{2}
\sum_{j=1}^{3}(\overline{\omega}^{2})^{j-1}a_{j3}^{k}\\
\sum_{j=1}^{3}(\overline{\omega}^{2})^{j-1}a_{j2}^{k}=\omega^{2}
\sum_{j=1}^{3}(\overline{\omega}^{2})^{j-1}a_{j3}^{k}\quad\\
\end{aligned}
\right..
\end{equation}
By Eqs. (9), (23), (24) and (25), we can get
\begin{equation}
\left\{
\begin{aligned}
\sum_{j=1}^{3}a_{j1}^{k}=a_{33}^{k}\qquad\quad\\
\sum_{j=1}^{3}\overline{\omega}^{j-1}a_{j1}^{k}= a_{33}^{k}\quad\\
\sum_{j=1}^{3}(\overline{\omega}^{2})^{j-1}a_{j1}^{k}=a_{33}^{k}
\end{aligned}
\right.
\end{equation}
and
\begin{equation}
\left\{
\begin{aligned}
\sum_{j=1}^{3}a_{j2}^{k}=a_{33}^{k}\qquad\quad\quad\\
\sum_{j=1}^{3}\overline{\omega}^{j-1}a_{j2}^{k}=\overline{\omega} a_{33}^{k}\quad\\
\sum_{j=1}^{3}(\overline{\omega}^{2})^{j-1}a_{j2}^{k}=\overline{\omega}^{2}a_{33}^{k}
\end{aligned}
\right..
\end{equation}
By Lemma 1-3, we can immediately get the unique solutions of Eqs. (26) and Eqs. (27),  \emph{i.e.},
\begin{equation}
\left\{
\begin{aligned}
a_{11}^{k}=a_{33}^{k}\quad\quad\\
a_{21}^{k}=a_{31}^{k}=0
\end{aligned}
\right.
\end{equation}
and
\begin{equation}
\left\{
\begin{aligned}
a_{22}^{k}=a_{33}^{k}\quad\quad\\
a_{12}^{k}=a_{32}^{k}=0
\end{aligned}
\right..
\end{equation}

By Eqs. (4), (5), (8), (9), (18), (19), (28) and (29), we have
\begin{equation}
\begin{split}
M_{k}^{\dag}M_{k}=
\left[
  \begin{array}{cccc}
    a_{33}^{k}   &0          &0            &0\\
    0            &a_{33}^{k} &0            &0\\
    0            &0          &a_{33}^{k}   &0\\
    0            &0          &0            &a_{33}^{k}\\
  \end{array}
\right]
\end{split}\nonumber
\end{equation}
for $k=1,2,\cdots,l.$ This means that all the POVM elements of Alice are proportional to identity matrix. Thus Alice cannot start with a nontrivial measurement to keep the orthogonality of the post-measurement states.

On the other hand, Bob faces the similar circumstance as Alice does since the set of these twelve states has a symmetrical structure. Therefore, these states cannot be exactly distinguished by using only LOCC. This completes the proof.
\end{proof}
\begin{figure}
\small
\centering
\includegraphics[scale=0.6]{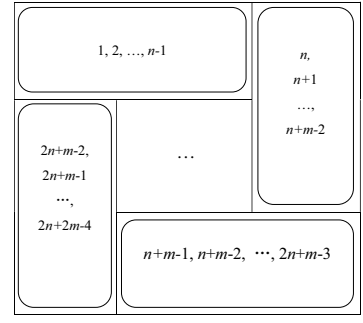}
\caption{Nonlocal subset of Shi \emph{et al.}'s UPB in $\mathbb{C}^{m} \otimes \mathbb{C}^{n}$ quantum system for $3\leq m\leq n$. Here $j$ denotes the product states $|\phi_{j}\rangle$ for $j$=$1$, $2$, $\cdots$, $2n+2m-4$. }
\end{figure}

A nonlocal subset of Shi \emph{et al.}'s UPB is given in Theorem 2. The structure of the subset can be seen in FIG. 2.

\begin{theorem}
In $\mathbb{C}^{m}\otimes\mathbb{C}^{n}$ quantum system, the set of the $2(m+n)-4$ product states in Eqs. (30) is nonlocal, \emph{i.e.}, it cannot be exactly discriminated using only LOCC.
\begin{equation}
\begin{aligned}
|\phi_{\sigma+1}\rangle=|0\rangle_{A}[\sum_{j=0}^{n-2}(\omega_{1})^{\sigma j}|j\rangle]_{B},\\
|\phi_{\eta+n}\rangle=[\sum_{j=0}^{m-2}(\omega_{2})^{\eta j}|j\rangle]_{A}|(n-1)\rangle_{B},\\
|\phi_{\sigma+n+m-1}\rangle=|(m-1)\rangle_{A}[\sum_{j=0}^{n-2}(\omega_{1})^{\sigma j}|(j+1)\rangle]_{B},\\
|\phi_{\eta+2n+m-2}\rangle=[\sum_{j=0}^{m-2}(\omega_{2})^{\eta j}|(j+1)\rangle]_{A}|0\rangle_{B},
\end{aligned}
\end{equation}
where $3\leq m \leq n$, $\omega_{1}=e^{\frac{2\pi \sqrt{-1}}{n-1}}$, $\omega_{2}=e^{\frac{2\pi \sqrt{-1}}{m-1}}$, $\sigma=0,1,\cdots,n-2$ and $\eta=0,1,\cdots,m-2.$
\end{theorem}

The proof of Theorem 2 is showed in Appendix A. It should be noted that Theorem 2 is still true when $m\geq 3$ and $n\geq3$, which can be easily seen from the proof process.
\section{Novel Nonlocal sets of orthogonal product states}

\begin{figure}
\small
\centering
\includegraphics[scale=0.6]{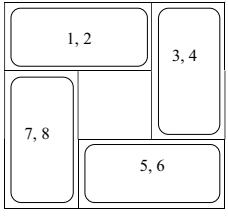}
\caption{The depiction of Bennett \emph{et al.}'s states in $\mathbb{C}^{3} \otimes \mathbb{C}^{3}$ quantum system as a set of dominoes. Here $j$ denotes the product state $|\phi_{j}\rangle$ for $j=1,\,2,\,\cdots,\,8.$}
\end{figure}
In Ref. \cite{Bennett1999}, Bennett \emph{et al.} constructed 9 orthogonal product states, which cannot be exactly discriminated by two separated observers if only LOCC are allowed between them. Feng \emph{et al.} \cite{Feng2009} showed that 8 (see FIG. 3 and Eqs. (31)) of those 9 product states are still indistinguishable using only LOCC in $\mathbb{C}^{3} \otimes \mathbb{C}^{3}$ quantum system.
\begin{equation}
\begin{aligned}
|\phi_{1}\rangle=|0\rangle_{A}(|0\rangle+|1\rangle)_{B},\\
|\phi_{2}\rangle=|0\rangle_{A}(|0\rangle-|1\rangle)_{B},\\
|\phi_{3}\rangle=(|0\rangle+|1\rangle)_{A}|2\rangle_{B},\\
|\phi_{4}\rangle=(|0\rangle-|1\rangle)_{A}|2\rangle_{B},\\
|\phi_{5}\rangle=|2\rangle_{A}(|1\rangle+|2\rangle)_{B},\\
|\phi_{6}\rangle=|2\rangle_{A}(|1\rangle-|2\rangle)_{B},\\
|\phi_{7}\rangle=(|1\rangle+|2\rangle)_{A}|2\rangle_{B},\\
|\phi_{8}\rangle=(|1\rangle-|2\rangle)_{A}|2\rangle_{B}.
\end{aligned}
\end{equation}

The set of these 8 states has a very good structure, which is very helpful for people to understand QNWE. Inspired by the construction method of these states, we give a nonlocal set of OPSs (see FIG. 4) in $\mathbb{C}^{5} \otimes \mathbb{C}^{5}$ quantum system, \emph{i.e.},

\begin{equation}
\begin{aligned}
|\phi_{1,2}\rangle=|0\rangle_{A}(|0\rangle\pm|1\rangle)_{B},\\
|\phi_{3,4}\rangle=|0\rangle_{A}(|2\rangle\pm|3\rangle)_{B},\\
|\phi_{5,6}\rangle=(|0\rangle\pm|1\rangle)_{A}|4\rangle_{B},\\
|\phi_{7,8}\rangle=(|2\rangle\pm|3\rangle)_{A}|4\rangle_{B},\\
|\phi_{9,10}\rangle=|4\rangle_{A}(|3\rangle\pm|4\rangle)_{B},\\
|\phi_{11,12}\rangle=|4\rangle_{A}(|1\rangle\pm|2\rangle)_{B},\\
|\phi_{13,14}\rangle=(|3\rangle\pm|4\rangle)_{A}|0\rangle_{B},\\
|\phi_{15,16}\rangle=(|1\rangle\pm|2\rangle)_{A}|4\rangle_{B}.\\
\end{aligned}
\end{equation}
\begin{theorem}
The set of 16 OPSs in Eqs. (32) is nonlocal, \emph{i.e.}, each state of the set cannot be perfectly distinguished by LOCC.
\end{theorem}
\begin{figure}
\small
\centering
\includegraphics[scale=0.6]{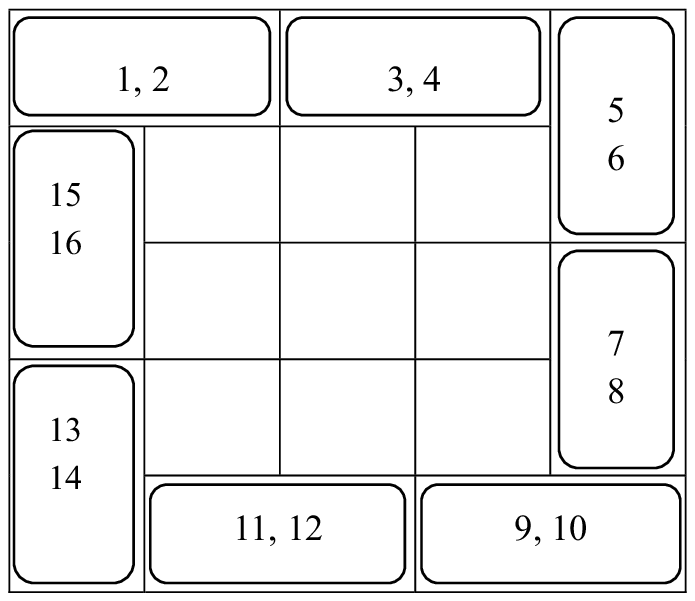}
\caption{A nonlocal set in $5\otimes5$ quantum system. $j$ denotes the product state $|\phi_{j}\rangle$ for $j=1,2,\cdots,16.$}
\end{figure}
\begin{proof} To discriminate these states, someone needs to perform a measurement of preserving orthogonality. That is, the states that are orthogonal only on Alice's (Bob's) side are still orthogonal on this side after measurement. Thus we need to show that Alice or Bob only can perform a nontrivial measurement no matter who goes first. Suppose that Alice performs a measure of preserving orthogonality with positive operator-value measurement (POVM) elements
\begin{equation}
\nonumber
\begin{split}
M_{k}^{\dag}M_{k}=
\left[
  \begin{array}{ccccc}
    a_{00}^{k}  &a_{01}^{k}  &a_{02}^{k}  &a_{03}^{k}  &a_{04}^{k}\\
    a_{10}^{k}  &a_{11}^{k}  &a_{12}^{k}  &a_{13}^{k}  &a_{14}^{k}\\
    a_{20}^{k}  &a_{21}^{k}  &a_{22}^{k}  &a_{23}^{k}  &a_{24}^{k}\\
    a_{30}^{k}  &a_{31}^{k}  &a_{32}^{k}  &a_{33}^{k}  &a_{34}^{k}\\
    a_{40}^{k}  &a_{41}^{k}  &a_{42}^{k}  &a_{43}^{k}  &a_{44}^{k}\\
\end{array}
\right]
\end{split}
\end{equation}
in the basis $\{|0\rangle,\,|1\rangle,\,|2\rangle,\,|3\rangle,\, |4\rangle\}$.

Because $(M_{k}\otimes I_{5\times5})|\phi_{1}\rangle$ is orthogonal to either of the states $\{(M_{k}\otimes I_{5\times5})|\phi_{j}\rangle,\,\,(M_{k}\otimes I)|\phi_{j+1}\rangle\}$ for $j=13,\,15$, we have
\begin{equation}
\nonumber
\left\{
\begin{aligned}
\langle\phi_{1}|M_{k}^{\dag}M_{k}\otimes I_{5\times5}|\phi_{j}\rangle=0\quad\\
\langle\phi_{1}|M_{k}^{\dag}M_{k}\otimes I_{5\times5}|\phi_{j+1}\rangle=0\\
\langle\phi_{j}|M_{k}^{\dag}M_{k}\otimes I_{5\times5}|\phi_{1}\rangle=0\quad\\
\langle\phi_{j+1}|M_{k}^{\dag}M_{k}\otimes I_{5\times5}|\phi_{1}\rangle=0\\
\end{aligned}
\right.
\end{equation}
Thus we get
\begin{equation}
\nonumber
\left\{
\begin{aligned}
\langle0|M_{k}^{\dag}M_{k}(|3\rangle+|4\rangle)_{A}=a^{k}_{03}+a^{k}_{04}=0\\
\langle0|M_{k}^{\dag}M_{k}(|3\rangle-|4\rangle)_{A}=a^{k}_{03}-a^{k}_{04}=0\\
(\langle3|+\langle4|)M_{k}^{\dag}M_{k}|0\rangle_{A}=a^{k}_{30}+a^{k}_{40}=0\\
(\langle3|-\langle4|)M_{k}^{\dag}M_{k}|0\rangle_{A}=a^{k}_{30}-a^{k}_{40}=0\\
\end{aligned}
\right.,
\end{equation}
and
\begin{equation}
\nonumber
\left\{
\begin{aligned}
\langle0|M_{k}^{\dag}M_{k}(|1\rangle+|2\rangle)_{A}=a^{k}_{01}+a^{k}_{02}=0\\
\langle0|M_{k}^{\dag}M_{k}(|1\rangle-|2\rangle)_{A}=a^{k}_{01}-a^{k}_{02}=0\\
(\langle1|+\langle2|)M_{k}^{\dag}M_{k}|0\rangle_{A}=a^{k}_{10}+a^{k}_{20}=0\\
(\langle1|-\langle2|)M_{k}^{\dag}M_{k}|0\rangle_{A}=a^{k}_{10}-a^{k}_{20}=0\\
\end{aligned}
\right..
\end{equation}
Therefore, we have
\begin{equation}
\left\{
\begin{aligned}
a^{k}_{01}=a^{k}_{02}=a^{k}_{03}=a^{k}_{04}=0\\
a^{k}_{10}=a^{k}_{20}=a^{k}_{30}=a^{k}_{40}=0\\
\end{aligned}
\right..
\end{equation}
Similarly, because $(M_{k}\otimes I_{5\times5})|\phi_{9}\rangle$ is orthogonal to either of the states $\{(M_{k}\otimes I_{5\times5})|\phi_{j}\rangle,\,\,(M_{k}\otimes I_{5\times5})|\phi_{j+1}\rangle\}$ for $j=5,\,7$, we have
\begin{equation}
\nonumber
\left\{
\begin{aligned}
\langle\phi_{9}|M_{k}^{\dag}M_{k}\otimes I_{5\times5}|\phi_{j}\rangle=0\quad\\
\langle\phi_{9}|M_{k}^{\dag}M_{k}\otimes I_{5\times5}|\phi_{j+1}\rangle=0\\
\langle\phi_{j}|M_{k}^{\dag}M_{k}\otimes I_{5\times5}|\phi_{9}\rangle=0\quad\\
\langle\phi_{j+1}|M_{k}^{\dag}M_{k}\otimes I_{5\times5}|\phi_{9}\rangle=0\\
\end{aligned}
\right..
\end{equation}
Thus we get
\begin{equation}
\nonumber
\left\{
\begin{aligned}
\langle4|M_{k}^{\dag}M_{k}(|0\rangle+|1\rangle)_{A}=a^{k}_{40}+a^{k}_{41}=0\\
\langle4|M_{k}^{\dag}M_{k}(|0\rangle-|1\rangle)_{A}=a^{k}_{40}-a^{k}_{41}=0\\
(\langle0|+\langle1|)M_{k}^{\dag}M_{k}|4\rangle_{A}=a^{k}_{04}+a^{k}_{14}=0\\
(\langle0|-\langle1|)M_{k}^{\dag}M_{k}|4\rangle_{A}=a^{k}_{04}-a^{k}_{14}=0\\
\end{aligned}
\right.
\end{equation}
and
\begin{equation}
\nonumber
\left\{
\begin{aligned}
\langle4|M_{k}^{\dag}M_{k}(|2\rangle+|3\rangle)_{A}=a^{k}_{42}+a^{k}_{43}=0\\
\langle4|M_{k}^{\dag}M_{k}(|2\rangle-|3\rangle)_{A}=a^{k}_{42}-a^{k}_{43}=0\\
(\langle2|+\langle3|)M_{k}^{\dag}M_{k}|4\rangle_{A}=a^{k}_{24}+a^{k}_{34}=0\\
(\langle2|-\langle3|)M_{k}^{\dag}M_{k}|4\rangle_{A}=a^{k}_{24}-a^{k}_{34}=0\\
\end{aligned}
\right.
\end{equation}
Therefore, we have
\begin{equation}
\left\{
\begin{aligned}
a^{k}_{40}=a^{k}_{41}=a^{k}_{42}=a^{k}_{43}=0\\
a^{k}_{04}=a^{k}_{14}=a^{k}_{24}=a^{k}_{34}=0\\
\end{aligned}
\right..
\end{equation}

Due to the orthogonality of $(M_{k}\otimes I_{5\times5})|\phi_{15}\rangle$ and  $(M_{k}\otimes I_{5\times5})|\phi_{16}\rangle$, we have
\begin{equation}
\nonumber
\left\{
\begin{aligned}
a_{11}-a_{12}+a_{21}-a_{22}=0\\
a_{11}+a_{12}-a_{21}-a_{22}=0\\
\end{aligned}
\right..
\end{equation}
Thus we get
\begin{equation}
\begin{aligned}
a_{11}=a_{22}.\\
\end{aligned}
\end{equation}
Similarly, we get
\begin{equation}
\left\{
\begin{aligned}
a_{33}=a_{44}\\
a_{00}=a_{11}\\
a_{22}=a_{33}\\
\end{aligned}
\right..
\end{equation}
by the orthogonality of $(M_{k}\otimes I_{5\times5})|\phi_{j}\rangle$ and  $(M_{k}\otimes I_{5\times5})|\phi_{j+1}\rangle$ for $j=13,\,5,\,7.$

Because $(M_{k}\otimes I_{5\times5})|\psi_{15}\rangle$ and $(M_{k}\otimes I_{5\times5})|\psi_{16}\rangle$ are orthogonal to $(M_{k}\otimes I_{5\times5})|\psi_{13}\rangle$ and $(M_{k}\otimes I_{5\times5})|\psi_{14}\rangle$ $\}$, we get
\begin{equation}
\nonumber
\left\{
\begin{aligned}
\langle\phi_{15}|M_{k}^{\dag}M_{k}\otimes I_{5\times5}|\phi_{13}\rangle=0\,\,\\
\langle\phi_{15}|M_{k}^{\dag}M_{k}\otimes I_{5\times5}|\phi_{14}\rangle=0\,\,\\
\langle\phi_{16}|M_{k}^{\dag}M_{k}\otimes I_{5\times5}|\phi_{13}\rangle=0\\
\langle\phi_{16}|M_{k}^{\dag}M_{k}\otimes I_{5\times5}|\phi_{14}\rangle=0\\
\end{aligned}
\right.,
\end{equation}
\begin{equation}
\nonumber
\left\{
\begin{aligned}
\langle\phi_{13}|M_{k}^{\dag}M_{k}\otimes I_{5\times5}|\phi_{15}\rangle=0\,\,\\
\langle\phi_{14}|M_{k}^{\dag}M_{k}\otimes I_{5\times5}|\phi_{15}\rangle=0\,\,\\
\langle\phi_{13}|M_{k}^{\dag}M_{k}\otimes I_{5\times5}|\phi_{16}\rangle=0\\
\langle\phi_{14}|M_{k}^{\dag}M_{k}\otimes I_{5\times5}|\phi_{16}\rangle=0\\
\end{aligned}
\right.,
\end{equation}
\emph{i.e.},
\begin{equation}
\left\{
\begin{aligned}
a_{13}^{k}+a_{14}^{k}+a_{23}^{k}+a_{24}^{k}=0\\
a_{13}^{k}-a_{14}^{k}+a_{23}^{k}-a_{24}^{k}=0\\
a_{13}^{k}+a_{14}^{k}-a_{23}^{k}-a_{24}^{k}=0\\
a_{13}^{k}-a_{14}^{k}-a_{23}^{k}+a_{24}^{k}=0\\
\end{aligned}
\right.,
\end{equation}
\begin{equation}
\left\{
\begin{aligned}
a_{31}^{k}+a_{32}^{k}+a_{41}^{k}+a_{42}^{k}=0\\
a_{31}^{k}+a_{32}^{k}-a_{41}^{k}-a_{42}^{k}=0\\
a_{31}^{k}-a_{32}^{k}+a_{41}^{k}-a_{42}^{k}=0\\
a_{31}^{k}-a_{32}^{k}-a_{41}^{k}+a_{42}^{k}=0\\
\end{aligned}
\right..
\end{equation}
By Eqs. (37) and (38), we have
\begin{equation}
\begin{aligned}
a_{13}^{k}=a_{14}^{k}=a_{23}^{k}=a_{24}^{k}=0,\\
\end{aligned}
\end{equation}
\begin{equation}
\begin{aligned}
a_{31}^{k}=a_{32}^{k}=a_{41}^{k}=a_{42}^{k}=0.\\
\end{aligned}
\end{equation}

Because $\{(M_{k}\otimes I_{5\times5})|\psi_{5}\rangle$ and $(M_{k}\otimes I_{5\times5})|\psi_{6}\rangle$ are orthogonal to $(M_{k}\otimes I)|\psi_{7}\rangle$ and $(M_{k}\otimes I)|\psi_{8}\rangle\}$ on Alice's side, \emph{i.e.},
\begin{equation}
\nonumber
\left\{
\begin{aligned}
\langle\phi_{5}|M_{k}^{\dag}M_{k}\otimes I_{5\times5}|\psi_{7}\rangle=0\\
\langle\phi_{5}|M_{k}^{\dag}M_{k}\otimes I_{5\times5}|\psi_{8}\rangle=0\\
\langle\phi_{6}|M_{k}^{\dag}M_{k}\otimes I_{5\times5}|\psi_{7}\rangle=0\\
\langle\phi_{6}|M_{k}^{\dag}M_{k}\otimes I_{5\times5}|\psi_{8}\rangle=0\\
\end{aligned}
\right.,
\end{equation}
\begin{equation}
\nonumber
\left\{
\begin{aligned}
\langle\phi_{7}|M_{k}^{\dag}M_{k}\otimes I_{5\times5}|\phi_{5}\rangle=0\\
\langle\phi_{7}|M_{k}^{\dag}M_{k}\otimes I_{5\times5}|\phi_{6}\rangle=0\\
\langle\phi_{8}|M_{k}^{\dag}M_{k}\otimes I_{5\times5}|\phi_{5}\rangle=0\\
\langle\phi_{8}|M_{k}^{\dag}M_{k}\otimes I_{5\times5}|\phi_{6}\rangle=0\\
\end{aligned}
\right.,
\end{equation}
we have
\begin{equation}
\left\{
\begin{aligned}
a_{02}^{k}+a_{03}^{k}+a_{12}^{k}+a_{13}^{k}=0\\
a_{02}^{k}-a_{03}^{k}+a_{12}^{k}-a_{13}^{k}=0\\
a_{02}^{k}+a_{03}^{k}-a_{12}^{k}-a_{13}^{k}=0\\
a_{02}^{k}-a_{03}^{k}-a_{12}^{k}+a_{13}^{k}=0\\
\end{aligned}
\right.,
\end{equation}
\begin{equation}
\left\{
\begin{aligned}
a_{20}^{k}+a_{21}^{k}+a_{30}^{k}+a_{31}^{k}=0\\
a_{20}^{k}-a_{21}^{k}+a_{30}^{k}-a_{31}^{k}=0\\
a_{20}^{k}+a_{21}^{k}-a_{30}^{k}-a_{31}^{k}=0\\
a_{20}^{k}-a_{21}^{k}-a_{30}^{k}+a_{31}^{k}=0\\
\end{aligned}
\right.,
\end{equation}
By Eqs.(41) and (42), we have
\begin{equation}
\begin{aligned}
a_{02}^{k}=a_{03}^{k}=a_{12}^{k}=a_{13}^{k}=0,\\
\end{aligned}
\end{equation}
\begin{equation}
\begin{aligned}
a_{20}^{k}=a_{21}^{k}=a_{30}^{k}=a_{31}^{k}=0.\\
\end{aligned}
\end{equation}

By Eqs. (33), (34), (35), (36), (39), (40), (43) and (44), we have
\begin{equation}
\nonumber
\begin{split}
M_{k}^{\dag}M_{k}=
\left[
  \begin{array}{ccccc}
    a_{00}^{k}   &0          &0            &0                 &0\\
    0            &a_{00}^{k} &0            &0                 &0\\
    0            &0          &a_{00}^{k}   &0                 &0\\
    0            &0          &0            &a_{00}^{k}        &0\\
    0            &0          &0            &0                 &a_{00}^{k} \\
  \end{array}
\right].
\end{split}
\end{equation}
That is, all the POVM elements of Alice are proportional to identity operator. This means that Alice cannot get any useful information to discriminate these states to preserve the orthogonality of the psot-measurement states. So does Bob by the symmetry of the set of OPSs in Eqs. (32). Therefore, these states cannot be reliably discriminated by LOCC, \emph{i.e.}, the set of these states is nonlocal. This completes the proof.
\end{proof}

By the structure of the set consisted of the OPSs in Theorem 3, we can easily generalize our construction to a general case.
\begin{figure}
\small
\centering
\includegraphics[scale=0.6]{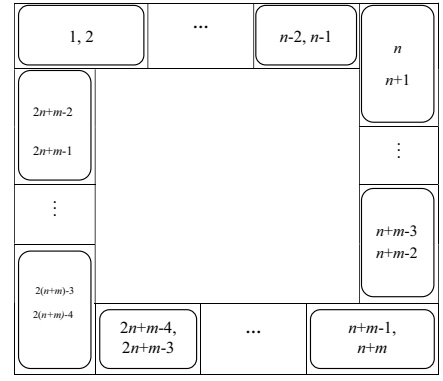}
\caption{A nonlocal set of $2(m+n)-4$ OPSs in $m\otimes n$ for $m=2d_{1}+1$ and $n=2d_{2}+1$.}
\end{figure}
\begin{theorem}
In $\mathbb{C}^{m}\otimes\mathbb{C}^{n}$ quantum system, the set of the following $2(m+n)-4$ OPSs, \emph{i.e.},
\begin{equation}
\begin{aligned}
|\phi_{1+2\delta}\rangle=|0\rangle_{A}[|2\delta\rangle+|(2\delta+1)\rangle]_{B},\\
|\phi_{2+2\delta}\rangle=|0\rangle_{A}[|2\delta\rangle-|(2\delta+1)\rangle]_{B},\\
|\phi_{n+2\sigma}\rangle=[|2\sigma\rangle+|(2\sigma+1)\rangle]_{A}|(n-1)\rangle_{B},\\
|\phi_{n+1+2\sigma}\rangle=[|2\sigma\rangle-|(2\sigma+1)\rangle]_{A}|(n-1)\rangle_{B},\\
|\phi_{n+m-1+2\delta}\rangle=|(m-1)\rangle_{A}[|(2\delta+1)\rangle+|(2\delta+2)\rangle]_{B},\\
|\phi_{n+m+2\delta}\rangle=|(m-1)\rangle_{A}[|(2\delta+1)\rangle-|(2\delta+2)\rangle]_{B},\\
|\phi_{2n+m-2+2\sigma}\rangle=[|(2\sigma+1)\rangle+|(2\sigma+2)\rangle]_{A}|0\rangle_{B},\\
|\phi_{2n+m-1+2\sigma}\rangle=[|(2\sigma+1)\rangle-|(2\sigma+2)\rangle]_{A}|0\rangle_{B},\\
\end{aligned}
\end{equation}
cannot be perfectly distinguished by LOCC, where $m$ and $n$ are positive integers, $m=2d_{1}+1\geq3$, $n=2d_{2}+1\geq3$; $\delta=0,$ $1,$ $2,$ $\cdots,$ $d_{2}-1$; and $\sigma=0,\,\, 1,\,\, 2,\,\, \cdots,\,\, d_{1}-1$. The structure of the states in Eqs. (45) is showed in FIG. 5.
\end{theorem}

\begin{proof}
We prove this Theorem by the same method that we use to prove Theorem 3. We need to show that Alice and Bob only can perform a trivial measurement to preserve the orthogonality of the post-measurement states no matter who goes first. Without loss of generality, suppose that Alice perform a general measurement with POVM element
\begin{equation}
\nonumber
\begin{split}
M_{k}^{\dag}M_{k}=
\left[
  \begin{array}{cccc}
    a_{00}^{k}             &a_{01}^{k}        &\cdots    &a_{0,m-1}^{k}\\
    a_{10}^{k}             &a_{11}^{k}        &\cdots    &a_{1,m-1}^{k}\\
    \vdots                 &\vdots            &\ddots    &\vdots \\
    a_{m-1,0}^{k}          &a_{m-1,1}^{k}     &\cdots    &a_{m-1,m-1}^{k} \\
\end{array}
\right]
\end{split}
\end{equation}
in the basis $\{|0\rangle,\,|1\rangle,\,\cdots,\,|(m-1)\rangle\}.$

Because $(M_{k}\otimes I_{n\times n} )|\phi_{1}\rangle$ is orthogonal to each state of the set
$\{(M_{k}\otimes I_{n\times n})|\phi_{2n+m-2+2\sigma}\rangle,\,
(M_{k}\otimes I_{n\times n})|\phi_{2n+m-1+2\sigma}\rangle\}$,
we have
\begin{equation}
\nonumber
\left\{
\begin{aligned}
\langle\phi_{1}|M_{k}^{\dag}M_{k}\otimes I_{n\times n}|\phi_{2n+m-2+2\sigma}\rangle=0\\
\langle\phi_{1}|M_{k}^{\dag}M_{k}\otimes I_{n\times n}|\phi_{2n+m-1+2\sigma}\rangle=0\\
\langle\phi_{2n+m-2+2\sigma}|M_{k}^{\dag}M_{k}\otimes I_{n\times n}|\phi_{1}\rangle=0\\
\langle\phi_{2n+m-1+2\sigma}|M_{k}^{\dag}M_{k}\otimes I_{n\times n}|\phi_{1}\rangle=0\\
\end{aligned}
\right.,
\end{equation}
where $\sigma=0,$ $1,$ $\cdots,$ $d_{1}-1$. Thus we have
\begin{equation}
\left\{
\begin{aligned}
a_{01}^{k}=a_{02}^{k}=a_{03}^{k}=\cdots=a_{0,m-1}^{k}=0\\
a_{10}^{k}=a_{20}^{k}=a_{30}^{k}=\cdots=a_{m-1,0}^{k}=0\\
\end{aligned}
\right..
\end{equation}
Similarly, we have
\begin{equation}
\left\{
\begin{aligned}
a_{0,m-1}^{k}=a_{1,m-1}^{k}=\cdots=a_{m-2,m-1}^{k}=0\\
a_{m-1,0}^{k}=a_{m-1,1}^{k}=\cdots=a_{m-1,m-2}^{k}=0\\
\end{aligned}
\right.
\end{equation}
by the orthogonality of $(M_{k}\otimes I_{n\times n})|\phi_{n+m-1}\rangle$ and each state of the set
$\{(M_{k}\otimes I_{n\times n})|\phi_{n+2\sigma}\rangle,\,
(M_{k}\otimes I_{n\times n})|\phi_{n+1+2\sigma}\rangle\}$.

Due to the orthogonality of $(M_{k}\otimes I_{n\times n})|\phi_{2n+m-2+2\sigma}\rangle$ and $(M_{k}\otimes I_{n\times n})|\phi_{2n+m-1+2\sigma}\rangle$, we have
\begin{equation}
\nonumber
\left\{
\begin{aligned}
\sum_{i=2\sigma+1}^{2\sigma+2}(-1)^{i+1}a_{2\sigma+1,i}^{k}+
\sum_{i=2\sigma+1}^{2\sigma+2}(-1)^{i+1}a_{2\sigma+2,i}^{k}=0\\
\sum_{i=2\sigma+1}^{2\sigma+2}a_{2\sigma+1,i}^{k}-
\sum_{i=2\sigma+1}^{2\sigma+2}a_{2\sigma+2,i}^{k}=0\qquad\qquad\qquad\,\,\\
\end{aligned}
\right..
\end{equation}
Therefore, we obtain
$$a_{2\sigma+1,2\sigma+1}^{k}=a_{2\sigma+2,2\sigma+2}^{k}$$
for $\sigma=0,$ $1,$ $2,$ $\cdots,$ $d_{1}-1$. Thus we have
\begin{equation}
a_{11}^{k}=a_{22}^{k}=\cdots=a_{m-1,m-1}^{k}.
\end{equation}
Similarly, due to the orthogonality of $(M_{k}\otimes I_{n\times n})|\phi_{n+2\sigma}
\rangle$ and $(M_{k}\otimes I_{n\times n})|\phi_{n+1+2\sigma}\rangle$, we have
\begin{equation}
\nonumber
\left\{
\begin{aligned}
a_{2\sigma,2\sigma}^{k}-a_{2\sigma,2\sigma+1}^{k}+a_{2\sigma+1,2\sigma}^{k}
-a_{2\sigma+1,2\sigma+1}^{k}=0\\
a_{2\sigma,2\sigma}^{k}+a_{2\sigma,2\sigma+1}^{k}-a_{2\sigma+1,2\sigma}^{k}-
a_{2\sigma+1,2\sigma+1}^{k}=0\\
\end{aligned}
\right..
\end{equation}
So
\begin{equation}
\nonumber
a_{2\sigma,2\sigma}^{k}=a_{2\sigma+1,2\sigma+1}^{k}\\
\end{equation}
where $\sigma=0,$ $1,$ $2,$ $\cdots,$ $d_{1}-1$. Thus we can get
\begin{equation}
a_{00}^{k}=a_{11}^{k}=\cdots=a_{m-2,m-2}^{k}.
\end{equation}

Because $(M_{k}\otimes I_{n\times n})|\phi_{2n+m-2+2\mu}\rangle$ and $(M_{k}\otimes$ $I_{n\times n})|\phi_{2n+m-1+2\mu}\rangle$ are orthogonal to $(M_{k}\otimes I_{n\times n})|\phi_{2n+m-2+2\lambda}\rangle$ and $(M_{k}\otimes I_{n\times n})|\phi_{2n+m-1+2\lambda}\rangle$, where $\mu,\,\lambda=0,$ $1,$ $2,$ $\cdots,$ $d_{1}-1$ and $\mu\neq \lambda$, \emph{i.e.},
\begin{equation}
\nonumber
\left\{
\begin{aligned}
\langle\phi_{2n+m-2+2\mu}|M_{k}^{\dag}M_{k}\otimes I_{n\times n}|\phi_{2n+m-2+2\lambda}
\rangle=0\\
\langle\phi_{2n+m-2+2\mu}|M_{k}^{\dag}M_{k}\otimes I_{n\times n}|\phi_{2n+m-1+2\lambda}
\rangle=0\\
\langle\phi_{2n+m-1+2\mu}|M_{k}^{\dag}M_{k}\otimes I_{n\times n}|\phi_{2n+m-2+2\lambda}
\rangle=0\\
\langle\phi_{2n+m-1+2\mu}|M_{k}^{\dag}M_{k}\otimes I_{n\times n}|\phi_{2n+m-1+2\lambda}
\rangle=0
\end{aligned}
\right.,
\end{equation}
so we have
\begin{equation}
\nonumber
\left\{
\begin{aligned}
\sum_{i=2\lambda+1}^{2\lambda+2}a_{2\mu+1,i}^{k}+
\sum_{i=2\lambda+1}^{2\lambda+2}a_{2\mu+2,i}^{k}=0\qquad\qquad\qquad\,\,\\
\sum_{i=2\lambda+1}^{2\lambda+2}(-1)^{i+1}a_{2\mu+1,i}^{k}+
\sum_{i=2\lambda+1}^{2\lambda+2}(-1)^{i+1}a_{2\mu+2,i}^{k}=0 \\
\sum_{i=2\lambda+1}^{2\lambda+2}a_{2\mu+1,i}^{k}-
\sum_{i=2\lambda+1}^{2\lambda+2}a_{2\mu+2,i}^{k}=0\qquad\qquad\qquad\,\,\\
\sum_{i=2\lambda+1}^{2\lambda+2}(-1)^{i+1}a_{2\mu+1,i}^{k}+
\sum_{i=2\lambda+1}^{2\lambda+2}(-1)^{i}a_{2\mu+2,i}^{k}=0\,\,\,\, \\
\end{aligned}
\right..
\end{equation}
Thus we get
\begin{equation}
\begin{aligned}
a_{2\mu+1,2\lambda+1}^{k}=a_{2\mu+1,2\lambda+2}^{k}=0,\\
a_{2\mu+2,2\lambda+1}^{k}=a_{2\mu+2,2\lambda+2}^{k}=0,
\end{aligned}
\end{equation}
where $\mu,\,\lambda=0,$ $1,$ $2,$ $\cdots,$ $d_{1}-1$ and $\mu\neq \lambda$. Similarly, because $(M_{k}\otimes I_{n\times n})|\phi_{n+2\mu}\rangle$ and $(M_{k}\otimes I_{n\times n})|\phi_{n+1+2\mu}\rangle$ are orthogonal to $(M_{k}\otimes I)|\phi_{n+2\lambda}\rangle$ and $(M_{k}\otimes I)|\phi_{n+1+2\lambda}\rangle$, we get
\begin{equation}
\begin{aligned}
a_{2\mu,2\lambda}^{k}=a_{2\mu,2\lambda+1}^{k}=0,\\
a_{2\mu+1,2\lambda}^{k}=a_{2\mu+1,2\lambda+1}^{k}=0.
\end{aligned}
\end{equation}
where $\mu,\,\lambda=0,$ $1,$ $2,$ $\cdots,$ $d_{1}-1$ and $\mu\neq \lambda$.

By Eqs. (46), (47), (48), (49), (50) and (51), we have
\begin{equation}
\nonumber
\begin{split}
M_{k}^{\dag}M_{k}=
\left[
  \begin{array}{cccc}
    a_{00}^{k}  &0           &\cdots    &0\\
    0           &a_{00}^{k}  &\cdots    &0\\
    \vdots      &\vdots      &\ddots    &\vdots \\
    0           &0           &\cdots    &a_{00}^{k} \\
\end{array}
\right]
\end{split}
\end{equation}
for $i=1,2,\cdots,l.$ This means that all the POVM elements of Alice are proportional to identity operator. Therefore, Alice only can perform a trivial measurement to preserve the orthogonality of the psot-measurement states. So does Bob by the symmetry of the set of these $2(m+n)-4$ product states. Therefore, these states cannot be reliably discriminated by LOCC no matter who goes first. This completes the proof.
\end{proof}

In Theorem 4, $m$ and $n$ are two odd numbers. In fact, a same conclusion can be got when $m$ and $n$ are all even numbers.
\begin{figure}
\small
\centering
\includegraphics[scale=0.6]{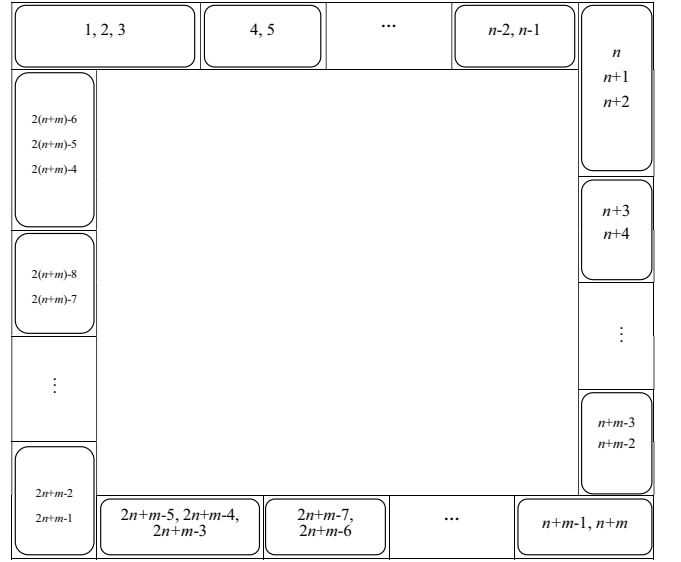}
\caption{A nonlocal set of $2(m+n)-4$ OPSs in $m\otimes n$ for $m=2d_{1}$ and $n=2d_{2}$.}
\end{figure}

\begin{theorem}
In $\mathbb{C}^{m}\otimes\mathbb{C}^{n}$ quantum system, the following $2(m+n)-4$ OPSs, \emph{i.e.},
\end{theorem}
\begin{equation}
\nonumber
\begin{aligned}
|\phi_{t+1}\rangle=|0\rangle_{A}(\sum_{j=0}^{2}\omega^{tj}|j\rangle)_{B},\\
|\phi_{2\delta+4}\rangle=|0\rangle_{A}[|(2\delta+3)\rangle+|(2\delta+4)\rangle]_{B},\\
|\phi_{2\delta+5}\rangle=|0\rangle_{A}[|(2\delta+3)\rangle-|(2\delta+4)\rangle]_{B},\\
|\phi_{t+n}\rangle=(\sum_{j=0}^{2}\omega^{tj}|j\rangle)_{A}|(n-1)\rangle_{B},\\
|\phi_{n+2\sigma+3}\rangle=[|(2\sigma+3)\rangle+|(2\sigma+4)\rangle]_{A}
|(n-1)\rangle_{B},\\
|\phi_{n+2\sigma+4}\rangle=[|(2\sigma+3)\rangle-|(2\sigma+4)\rangle]_{A}
|(n-1)\rangle_{B},\\
|\phi_{t+n+m-1}\rangle=|(m-1)\rangle_{A}[\sum_{j=0}^{2}\omega^{tj}
|(j+1)\rangle]_{B},\\
|\phi_{n+m+2\delta+2}\rangle=|(m-1)\rangle_{A}[|(2\delta+4)\rangle
+|(2\delta+5)\rangle]_{B},\\
|\phi_{n+m+2\delta+3}\rangle=|(m-1)\rangle_{A}[|(2\delta+4)\rangle
-|(2\delta+5)\rangle]_{B},\\
|\phi_{t+2n+m-2}\rangle=[\sum_{j=0}^{2}\omega^{tj}|(j+1)\rangle]_{A}
|0\rangle_{B},\\
|\phi_{2n+m+2\sigma+1}\rangle=[|(2\sigma+4)\rangle+|(2\sigma+5)\rangle]_{A}
|0\rangle_{B},\\
|\phi_{2n+m+2\sigma+2}\rangle=[|(2\sigma+4)\rangle-|(2\sigma+5)\rangle]_{A}
|0\rangle_{B},\\
\end{aligned}
\end{equation}
cannot be perfectly distinguished by LOCC, where $m=2d_{1}\geq4$ and $n=2d_{2}\geq4$; $t=0,\,\,1,\,\,2$; $\omega=e^{\frac{2\pi \sqrt{-1}}{3}}$; \,\,$\delta=0,$ $1,$ $2,$ $\cdots,$ $d_{2}-3$; and $\sigma=0,\,\, 1,\,\, 2,\,\, \cdots,\,\, d_{1}-3$. The structure of these states is showed in FIG. 6.

The proof of Theorem 5 is given in Appendix B. By Theorem 5, we can get a special set for $m=n=4$, \emph{i.e.},
\begin{equation}
\begin{aligned}
|\phi_{t+1}\rangle=|0\rangle_{A}(\sum_{j=0}^{2}\omega^{tj}|j\rangle)_{B},\\
|\phi_{t+4}\rangle=(\sum_{j=0}^{2}\omega^{tj}|j\rangle)_{A}|3\rangle_{B},\\
|\phi_{t+7}\rangle=|3\rangle_{A}[\sum_{j=0}^{2}\omega^{tj}
|(j+1)\rangle]_{B},\\
|\phi_{t+10}\rangle=[\sum_{j=0}^{2}\omega^{tj}|(j+1)\rangle]_{A}
|0\rangle_{B},\\
\end{aligned}
\end{equation}
where $t=0,\,\,1,\,\,2$; $\omega=e^{\frac{2\pi \sqrt{-1}}{3}}$. It is easy to see that the states in Eqs. (52) are identical to the states in Eqs. (1).

\begin{theorem}
In $\mathbb{C}^{m}\otimes\mathbb{C}^{n}$ quantum system, the following $2(m+n)-4$ OPSs
\begin{equation}
\nonumber
\begin{aligned}
|\psi_{1+2\delta}\rangle=|0\rangle_{A}[|2\delta\rangle+|(2\delta+1)\rangle]_{B},\\
|\psi_{2+2\delta}\rangle=|0\rangle_{A}[|2\delta\rangle-|(2\delta+1)\rangle]_{B},\\
|\psi_{t+n}\rangle=(\sum_{j=0}^{2}\omega^{tj}|j\rangle)_{A}|(n-1)\rangle_{B},\\
|\psi_{n+2\sigma+3}\rangle=[|(2\sigma+3)\rangle+|(2\sigma+4)\rangle]_{A}
|(n-1)\rangle_{B},\\
|\psi_{n+2\sigma+4}\rangle=[|(2\sigma+3)\rangle-|(2\sigma+4)\rangle]_{A}
|(n-1)\rangle_{B},\\
|\psi_{n+m+2\delta-1}\rangle=|(m-1)\rangle_{A}[|(2\delta+1)\rangle+|(2\delta+2)\rangle]_{B},\\
|\psi_{n+m+2\delta}\rangle=|(m-1)\rangle_{A}[|(2\delta+1)\rangle-|(2\delta+2)\rangle]_{B},\\
|\psi_{t+2n+m-2}\rangle=[\sum_{j=0}^{2}\omega^{tj}|(j+1)\rangle]_{A}|0\rangle_{B},\\
|\psi_{2n+m+2\sigma+1}\rangle=[|(2\sigma+4)\rangle+|(2\sigma+5)\rangle]_{A}
|0\rangle_{B},\\
|\psi_{2n+m+2\sigma+2}\rangle=[|(2\sigma+4)\rangle-|(2\sigma+5)\rangle]_{A}
|0\rangle_{B},\\
\end{aligned}
\end{equation}
cannot be perfectly distinguished by LOCC, where $m=2d_{1}\geq4$ and $n=2d_{2}+1\geq3$; $\omega=e^{\frac{2\pi \sqrt{-1}}{3}}$, $t=0,\,\,1,\,\,2;$ $\delta=0,$ $1,$ $2,$ $\cdots,$ $d_{2}-1$; and $\sigma=0,\,\, 1,\,\, 2,\,\, \cdots,\,\, d_{1}-3$.
\end{theorem}
\begin{figure}
\small
\centering
\includegraphics[scale=0.6]{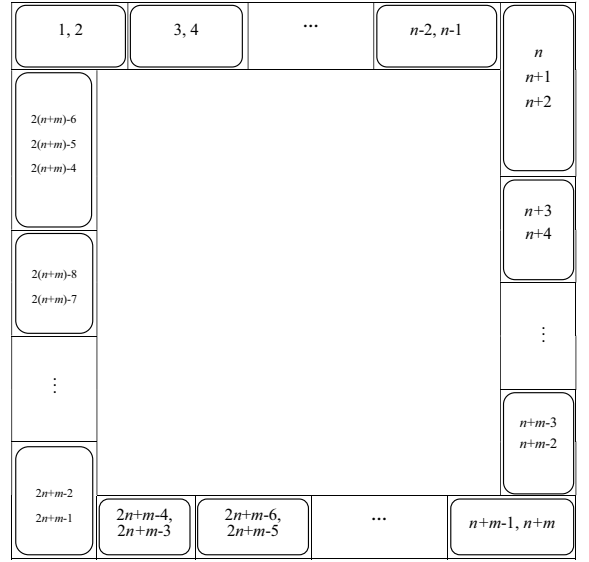}
\caption{A nonlocal set of $2(m+n)-4$ OPSs in $m\otimes n$ for $m=2d_{1}$ and $n=2d_{2}+1$. }
\end{figure}

The structure of the states in Theorem 6 is showed in FIG. 7. We can prove Theorem 7 with the proof methods of Theorem 4 and Theorem 5. That is, we prove Alice only can perform a trivial measurement as we do in Theorem 5 while we prove Bob only can perform a trivial measurement as we do in Theorem 4. By Eqs. (31), Theorem 1 and Theorem 3-6, we get a general conclusion, \emph{i.e.}, Theorem 7.
\begin{theorem}
In $\mathbb{C}^{m}\otimes\mathbb{C}^{n}$ quantum system, there exists $2(m+n)-4$ OPSs that cannot be perfectly discriminated by LOCC, where $m\geq 3$ and $n\geq3$.
\end{theorem}

In fact, the sets of OPSs in Theorem 2, 4, 5 or 6 can be extended to complete orthogonal product bases by adding the OPSs $|i\rangle|j\rangle$, where $1\leq i\leq m-2$ and $1\leq j\leq n-2$. This means that these sets are completable. For bipartite systems, there are many constructions of locally indistinguishable sets of orthogonal product states. Some of them are complete bases \cite{Bennett1999,Zhang2014}, while some of them are incomplete \cite{CHB1999,DiVincenzo2003,SHalder2019,Fei2020,Sixia2015,
Wang2015,Zhang2015,XZhang2016,Wang2015,Zhang2016,Zhangxiaoqian2017,WangYLING2017}. Among the incomplete sets, several sets are completable, a few sets are uncompletable, and rest of them are unextendible product bases. TABLE \uppercase\expandafter{\romannumeral1} gives the detailed indicators of different sets.

Our nonlocal sets of OPSs have fewer elements than the existing results \cite{XZhang2016,Zhangxiaoqian2017,Wang2015,Zhang2016} in a general $\mathbb{C}^{m}\otimes\mathbb{C}^{n}$ quantum system
when a ``stoper" state \cite{DiVincenzo2003} is not include. That is, we give the minimum number of elements to form a completable set of OPSs that cannot be perfectly distinguished by LOCC in arbitrary bipartite system. In fact, $2(m+n)-4$ is just the minimum number to form a completable and nonlocal set of OPSs in $\mathbb{C}^{m}\otimes\mathbb{C}^{n}$ quantum system for $m\geq 3$ and $n\geq 3$ so far. Most importantly, the structure of our nonlocal set of OPSs is symmetrical and more intuitive.

\begin{widetext}
\begin{center}
\begin{table}[htpp]
\centering
\caption{Comparison of different constructions of locally indistinguishable sets of bipartite orthogonal product states.}
\label{tab:ModelASymbol}
\begin{tabular}{cccccc}
\hline

Different sets           &The number of elements            &system  &Constraints of parameters   &type 1 &type 2 \\
\hline
Construction of Ref. \cite{Bennett1999}   &9               &$\mathbb{C}^{3}\otimes\mathbb{C}^{3}$   &- &complete  &- \\
Construction of Ref. \cite{Zhang2014} & $d^{2}$  &$\mathbb{C}^{d}\otimes\mathbb{C}^{d}$   &$d$ is odd and $d\geq 3$  &complete  &-\\
Construction of Ref. \cite{CHB1999}  & 3  &$\mathbb{C}^{3}\otimes\mathbb{C}^{3}$ &- &incomplete &unexpendible\\
Construction of Ref. \cite{DiVincenzo2003} &$d^{2}-2d+1$   &$\mathbb{C}^{d}\otimes\mathbb{C}^{d}$ &$d$ is even and $d\geq 4$ &incomplete &unexpendible\\
Construction of Ref. \cite{SHalder2019} &$d^{2}-2d+2$ &$\mathbb{C}^{d}\otimes\mathbb{C}^{d}$ &$d$ is odd and $d\geq 3$ &incomplete &unexpendible\\
Construction of Ref. \cite{Fei2020} &$mn-4\lfloor\frac{m-1}{2}\rfloor$ &$\mathbb{C}^{m}\otimes\mathbb{C}^{n}$ &$3\leq m \leq n$ &incomplete &unexpendible\\
Construction of Ref. \cite{Sixia2015}    &2$d$-1 &$\mathbb{C}^{d}\otimes\mathbb{C}^{d}$   &$d\in Z$ and $d\geq 3$        &incomplete &uncompletable \\
Construction of Ref. \cite{Wang2015} &$6d-9$ &$\mathbb{C}^{d}\otimes\mathbb{C}^{d}$   &$d$ is odd and $d\geq 3$ &incomplete &completable  \\
Construction of Ref. \cite{Zhang2015} &$4d-4$ &$\mathbb{C}^{d}\otimes\mathbb{C}^{d}$   &$d\in Z$ and $d\geq 3$ &incomplete &completable\\
Construction of Ref. \cite{XZhang2016}          &$mn$  &$\mathbb{C}^{m}\otimes\mathbb{C}^{n}$   &$m\geq 3$ and $n\geq 3$  &incompete &completable\\
 Construction of Ref. \cite{Wang2015}                &$3m+3n-9$   &$\mathbb{C}^{m}\otimes\mathbb{C}^{n}$   &$m\geq 3$ and $n\geq 3$  &incomplete &completable\\
Construction of Ref. \cite{Zhang2016}               &$3n+m-4$    &$\mathbb{C}^{m}\otimes\mathbb{C}^{n}$   &$3\leq m\leq n$ &incomplete &completabe \\
Construction of Ref. \cite{Zhang2016}               &$2n-1$    &$\mathbb{C}^{m}\otimes\mathbb{C}^{n}$   &$3\leq m\leq n$ &incomplete &uncompletabe \\
Construction of Ref. \cite{WangYLING2017}               &$2n-1$    &$\mathbb{C}^{m}\otimes\mathbb{C}^{n}$   &$4\leq m\leq n$ &incomplete &uncompletabe \\
Construction of Ref. \cite{Zhangxiaoqian2017}       &$3m+3n-8$   &$\mathbb{C}^{m}\otimes\mathbb{C}^{n}$  &$4\leq m \leq n$  &incomplete &uncompletabe\\
Our sets in Theorem 4-6                        &$2n+2m-4$   &$\mathbb{C}^{m}\otimes\mathbb{C}^{n}$   &$m\geq3$ and $n\geq 3$  &incomplete &completable\\
\hline
\end{tabular}
\end{table}
\end{center}
\end{widetext}

\section{isomorphism of two sets of orthogonal product states}

In Ref. \cite{DiVincenzo2003}, DiVincenzo \emph{et al.} gave the concept of orthogonality graph to describe the structure of a set of orthogonal product states in a bipartite Hilbert space.
\begin{definition}\cite{DiVincenzo2003} Let $H=H_{A}\otimes H_{B}$ is a bipartite Hilbert space with dim $H_{A}$=dim $H_{B}$. Let $S=\{|\psi_{j}\rangle \equiv |\varphi_{Aj}\rangle\otimes|\varphi_{Bj}\rangle|$$j=1$, $2$, $\cdots$, $s\}$ is a set of orthogonal product states in $H$. We represent $S$ as a graph $G=(V, E_{A}\cup E_{B})$, where $E_{A}$ and $E_{B}$ are the sets of edges. Each state $|\psi_{j}\rangle\in S$ is represented as a vertex of $V$. If the states $|\psi_{t}\rangle$ and $|\psi_{i}\rangle$ are orthogonal on $H_{A}$ $(H_{B})$, there will be a solid (dotted) line between the vertices $v_{t}$ and $v_{i}$.
\end{definition}
Now we give the definition of isomorphism of two sets, which will be used to describe the relations between different sets of locally indistinguishable OPSs.
\begin{definition}
In $\mathbb{C}^{m}\otimes \mathbb{C}^{n}$ quantum system, we say two sets of orthogonal product states are isomorphic if their orthogonality graphs are same both on Alice's side and Bob's side; Otherwise, we say these two sets are not isomorphic.
\end{definition}
The orthogonality graph of the states in Eqs. (31) on Alice's side is showed in FIG. 8. In FIG. 8, the vertex $j$ denotes the sates $|\phi_{j}\rangle$ for $j=1,2,\cdots,8.$
\begin{figure}
\small
\centering
\includegraphics[scale=0.6]{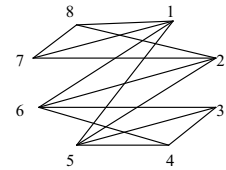}
\caption{The orthogonality graph of the states in Eqs. (31) on Alice's side.}
\end{figure}

  In Refs. \cite{Zhang2015}, Zhang \emph{et al.} constructed a nonlocal set of OPSs in
$\mathbb{C}^{d}\otimes\mathbb{C}^{d}$ quantum system. When $d=3$, we can get a nonlocal set of OPSs in $\mathbb{C}^{3}\otimes\mathbb{C}^{3}$ quantum system, \emph{i.e.},
\begin{equation}
\begin{aligned}
|\psi_{1}\rangle=|1\rangle_{A}(|0\rangle+|1\rangle)_{B}\\
|\psi_{2}\rangle=|1\rangle_{A}(|0\rangle-|1\rangle)_{B}\\
|\psi_{3}\rangle=|2\rangle_{A}(|0\rangle+|2\rangle)_{B}\\
|\psi_{4}\rangle=|2\rangle_{A}(|0\rangle-|2\rangle)_{B}\\
|\psi_{5}\rangle=(|0\rangle+|1\rangle)_{A}|2\rangle_{B}\\
|\psi_{6}\rangle=(|0\rangle-|1\rangle)_{A}|2\rangle_{B}\\
|\psi_{7}\rangle=(|0\rangle+|2\rangle)_{A}|1\rangle_{B}\\
|\psi_{8}\rangle=(|0\rangle-|2\rangle)_{A}|1\rangle_{B}
\end{aligned}.
\end{equation}
\begin{figure}
\small
\centering
\includegraphics[scale=0.6]{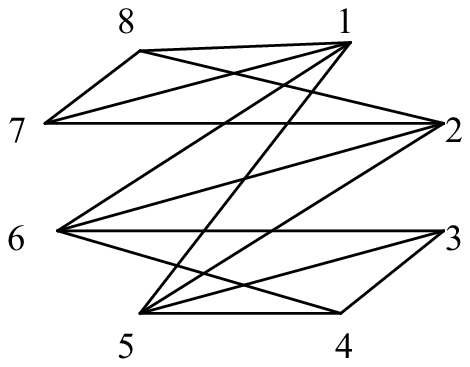}
\caption{The orthogonality graph of the states in Eqs. (54) on Alice's side.}
\end{figure}
Here we give a mapping between the states in Eqs. (53) and the serial numbers of the vertices, \emph{i.e.,}
\begin{equation}
\begin{aligned}
|\psi_{1}\rangle\mapsto1,&&|\psi_{2}\rangle\mapsto2,&&|\psi_{5}\rangle\mapsto3,&&
|\psi_{6}\rangle\mapsto4,\\
|\psi_{3}\rangle\mapsto5,&&|\psi_{4}\rangle\mapsto6,&&|\psi_{7}\rangle\mapsto7,&&
|\psi_{8}\rangle\mapsto8.\\
\end{aligned}
\end{equation}

\begin{figure}
\small
\centering
\includegraphics[scale=0.6]{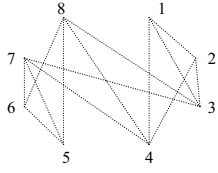}
\caption{The orthogonality graph of the states in Eqs. (31) or Eqs. (53) on Bob's side.}
\end{figure}

By the mapping relation in Eqs. (54), we give the orthogonality graph of the states in Eqs. (53) on Alice's subsystem (see FIG. 9). It is easy to see that FIG. 8 and FIG. 9 are identical. This means that the two sets of orthogonal product states have the same orthogonal graph on Alice's side. Similarly, we can find the set of the states in Eqs. (31) and the set of the states in Eqs. (53) have the same orthogonal graph (see FIG. 10) on Bob's side. Therefore, these two sets are isomorphic.
\begin{figure}
\small
\centering
\includegraphics[scale=0.6]{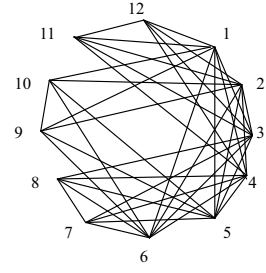}
\caption{The orthogonality graph of Zhang \emph{et al}.'s states on Alice's side.}
\end{figure}
\begin{figure}
\small
\centering
\includegraphics[scale=0.6]{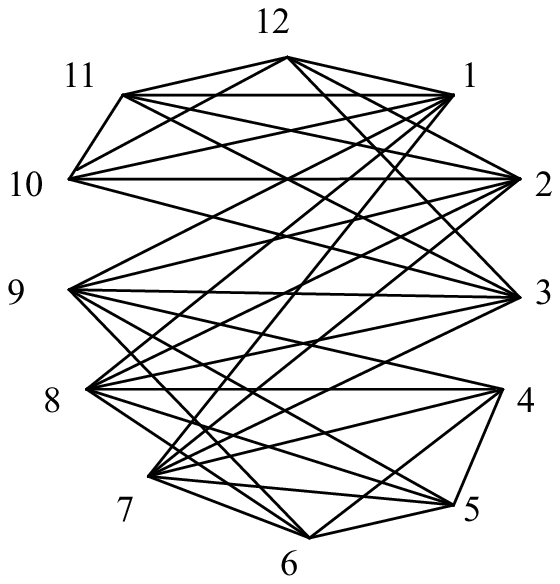}
\caption{The orthogonality graph of our states in Eqs. (2) on Alice's side.}
\end{figure}
\begin{figure}
\small
\centering
\includegraphics[scale=0.6]{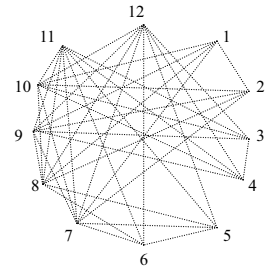}
\caption{The orthogonality graph of zhang \emph{et al's} states on Bob's side.}
\end{figure}
\begin{figure}
\small
\centering
\includegraphics[scale=0.6]{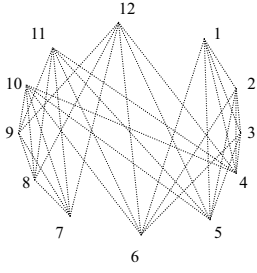}
\caption{The orthogonality graph of our states on Bob's side.}
\end{figure}

Now we consider the orthogonality graphs of Zhang \emph{et al.}'s set and ours in $\mathbb{C}^{4}\otimes\mathbb{C}^{4}$ quantum system. The states constructed by Zhang \emph{et al.}'s method are as follows (see Eq. (55)).
\begin{equation}
\begin{aligned}
|\psi_{1}\rangle=|1\rangle_{A}(|0\rangle+|1\rangle)_{B},\\
|\psi_{2}\rangle=|1\rangle_{A}(|0\rangle-|1\rangle)_{B},\\
|\psi_{3}\rangle=|2\rangle_{A}(|0\rangle+|2\rangle)_{B},\\
|\psi_{4}\rangle=|2\rangle_{A}(|0\rangle-|2\rangle)_{B},\\
|\psi_{5}\rangle=|3\rangle_{A}(|0\rangle+|3\rangle)_{B},\\
|\psi_{6}\rangle=|3\rangle_{A}(|0\rangle-|3\rangle)_{B},\\
|\psi_{7}\rangle=(|0\rangle+|1\rangle)_{A}|2\rangle_{B},\\
|\psi_{8}\rangle=(|0\rangle-|1\rangle)_{A}|2\rangle_{B},\\
|\psi_{9}\rangle=(|0\rangle+|2\rangle)_{A}|3\rangle_{B},\\
|\psi_{10}\rangle=(|0\rangle-|2\rangle)_{A}|3\rangle_{B},\\
|\psi_{11}\rangle=(|0\rangle+|3\rangle)_{A}|1\rangle_{B},\\
|\psi_{12}\rangle=(|0\rangle-|3\rangle)_{A}|1\rangle_{B}.
\end{aligned}
\end{equation}
The orthogonality graphs of the states constructed by Zhang \emph{et al.}'s method are exhibited in FIG. 11 and FIG. 13. The orthogonality graphs of our construction, \emph{i.e.}, the states in Eqs. (52), are showed in FIG. 12 and FIG. 14. Obviously, FIG. 11 has 39 edges while FIG. 12 has 33 edges. This means that the orthogonality graph of Zhang \emph{et al.}'s states in $\mathbb{C}^{4}\otimes \mathbb{C}^{4}$ quantum system and ours are not identical on Alice's side. So do they on Bob's side. Thus the set of the states constructed by Zhang \emph{et al.} is not isomorphic with the states in Eqs. (52). But why does this happen? This is because any two states of our set in $\mathbb{C}^{4}\otimes \mathbb{C}^{4}$ quantum system are orthogonal only on one side while some states of Zhang \emph{et al.} are orthogonal both on Alice's side and on Bob's side. For example, $|\psi_{3}\rangle$ and $|\psi_{11}\rangle$, $|\psi_{5}\rangle$ and $|\psi_{7}\rangle$ in Eqs. (56).

By Theorem 2 and Theorem 4, we can know that the subset of Shi \emph{et al}.'s UPB is same as the set constructed by our novel method in $\mathbb{C}^{3}\otimes \mathbb{C}^{3}$ quantum system. On the other hand, the subset (See Eqs. (1)) of Shi \emph{et al}.'s UPB in $\mathbb{C}^{4}\otimes \mathbb{C}^{4}$ is same as ours (See Eqs.(52)). Thus the subset of Shi \emph{et al}.'s UPB is isomorphic to ours both in $\mathbb{C}^{3}\otimes \mathbb{C}^{3}$ and $\mathbb{C}^{4}\otimes \mathbb{C}^{4}$. Now we consider the case in $\mathbb{C}^{5}\otimes \mathbb{C}^{5}$. By Theorem 2, we get the subset of Shi \emph{et al}'s UPB as follow, \emph{i.e.},

\begin{equation}
\begin{aligned}
|\phi_{t+1}\rangle=|0\rangle_{A}[\sum_{j=0}^{3}\omega^{t j}|j\rangle]_{B},\\
|\phi_{t+5}\rangle=[\sum_{j=0}^{3}\omega^{t j}|j\rangle]_{A}|4\rangle_{B},\\
|\phi_{t+9}\rangle=|4\rangle_{A}[\sum_{j=0}^{3}\omega^{t j}|(j+1)\rangle]_{B},\\
|\phi_{t+13}\rangle=[\sum_{j=0}^{3}\omega^{t j}|(j+1)\rangle]_{A}|0\rangle_{B},
\end{aligned}
\end{equation}
where $m\geq3$, $n\geq3$, $\omega=e^{\frac{2\pi \sqrt{-1}}{4}}$, $t=0,\,1,\,2,\,3$. The orthogonality graph of these states on Alice's side is shown in FIG. 15.
\begin{figure}
\small
\centering
\includegraphics[scale=0.6]{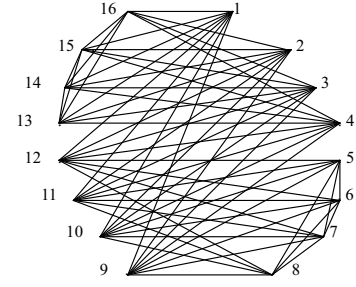}
\caption{The orthogonality graph of our states on Alice's side.}
\end{figure}
The orthogonality graph of our novel states in Eqs. (32) on Alice's side is shown in FIG. 16.
\begin{figure}
\small
\centering
\includegraphics[scale=0.6]{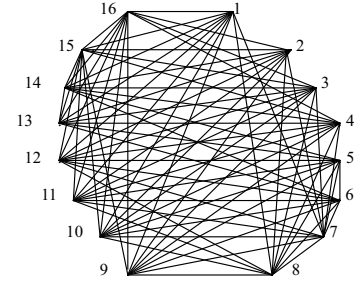}
\caption{The orthogonality graph of the states in Eqs. (57) on Alice's side.}
\end{figure}
Since FIG. 15 and FIG. 16 are different, we know that the two sets corresponding to Fig. 15 and Fig. 16 are not isomorphic. That is, the subset of Shi \emph{et al}.'s UPB and the set constructed by our novel method are not isomorphic in $\mathbb{C}^{5} \otimes \mathbb{C}^{5}$ quantum system. The same case exists in $\mathbb{C}^{m} \otimes \mathbb{C}^{n}$ quantum system for $m\geq 6$ and $n\geq 6$.

\section{Conclusion}

QNWE is a peculiar phenomenon in nature. Many people are devoted to this research work \cite{Xin2008,MA2014,Chen2004,Rinaldis,Zhangzc2015,Bravyi2004,Johnston2014,Chen2015,
Nathaniel2013} since it is a significant thing to improve quantum information theory. In addition, many results are used to design quantum cryptographic protocols \cite{Yang2015,DHJQ2020,Rahaman2015}.

It is interesting to find out the minimum number of elements to form a nonlocal set of OPSs. We found a subset of Shi \emph{et al.}'s UPB, which only has $2(m+n)-4$ member, cannot be perfectly distinguished by LOCC in $\mathbb{C}^{m}\otimes\mathbb{C}^{n}$ quantum system. On the other hand, we give a novel method to construct a nonlocal set
in $\mathbb{C}^{m}\otimes\mathbb{C}^{n}$ quantum system for $m\geq 3$ and $n\geq 3$.
As we know, $2(m+n)-4$ is the minimum number of elements to form a completable set of OPSs that cannot be perfectly distinguished by LOCC in $\mathbb{C}^{m}\otimes\mathbb{C}^{n}$ quantum system for $m\geq 3$ and $n\geq 3$. Furthermore, we give the concept of isomorphism for two nonlocal sets of OPSs.
We analyze the structures of two sets of OPSs from the perspective of isomorphism. All the results are useful to perfect the theory of quantum nonlocality.
\begin{acknowledgments}The authors are grateful for the anonymous referee's
suggestions to improve the quality of this paper. This work is supported by Natural Science Foundation of Shandong province of China (Grants No. ZR2019MF023), Open Foundation of State Key Laboratory of Networking
and Switching Technology (Beijing University of Posts and Telecommunications)
(SKLNST-2019-2-01) and SDUST Research Fund.
\end{acknowledgments}
\appendix
\section{Proof of Theorem 2}
\begin{proof}
Suppose that Alice firstly performs a POVM measurement with a set of general $m\times m$ POVM elements $\{M_{k}^{\dag}M_{k}: k=1,\,2,\,\cdots,\,l\}$, where
  \begin{equation}
\begin{split}
M_{k}^{\dag}M_{k}=
\left[
  \begin{array}{ccccc}
    a_{00}^{k}      &a_{01}^{k}    &\cdots &a_{0(m-1)}^{k}\\
    a_{10}^{k}     &a_{11}^{k}     &\cdots &a_{1(m-1)}^{k}\\
    \vdots         &\vdots         &\ddots &\vdots\\
    a_{(m-1)0}^{k} &a_{(m-1)1}^{k} &\cdots &a_{(m-1)(m-1)}^{k}\\
  \end{array}
\right]
\end{split}\nonumber
\end{equation}
in the basis $\{|0\rangle,$ $|1\rangle,$ $\cdots,$ $|(m-1)\rangle\}$. The post-measurement states $\{(M_{k}\otimes I_{n\times n})|\phi_{j}\rangle:$ $j=1,$ $2$, $\cdots$, $2(m+n)-4\}$ should be mutually orthogonal to make further discrimination possible.

Because $(M_{k}\otimes I_{n\times n})|\phi_{1}\rangle$ should be orthogonal to $(M_{k}\otimes I_{n\times n})|\phi_{\eta+2n+m-2}\rangle$ on Alice's side for $\eta=0,\,1,\,\cdots,\,m-2$, we get
\begin{equation}
\left\{
\begin{aligned}
\sum_{j=1}^{m-1}a_{0j}^{k}=0\qquad\qquad\quad\\
\sum_{j=1}^{m-1}(\omega_{2})^{j-1} a_{0j}^{k}=0\qquad\\
\sum_{j=1}^{m-1}[(\omega_{2})^{2}]^{j-1} a_{0j}^{k}=0\quad\\
\vdots\,\,\,\,\,\,\,\,\,\,\,\,\,\,\,\,\,\,\,\,\,\,\,\,\,\,\,\\
\sum_{j=1}^{m-1}[(\omega_{2})^{m-2}]^{j-1} a_{0j}^{k}=0\\
\end{aligned}
\right.
\end{equation}
and
\begin{equation}
\left\{
\begin{aligned}
\sum_{j=1}^{m-1}a_{j0}^{k}=0\qquad\qquad\,\,\,\\
\sum_{j=1}^{m-1}\overline{\omega_{2}}^{j-1}a_{j0}^{k}=0\qquad\\
\sum_{j=1}^{m-1}(\overline{\omega_{2}}^{2})^{j-1}a_{j0}^{k}=0\quad\\
\vdots\,\,\,\,\,\,\,\,\,\,\,\,\,\,\,\,\,\,\,\,\,\,\,\,\,\,\,\,\\
\sum_{j=1}^{m-1}(\overline{\omega_{2}}^{m-2})^{j-1}a_{j0}^{k}=0\\
\end{aligned}
\right.,
\end{equation}
where $\overline{\omega_{2}}$ is the conjugate complex number of $\omega_{2}$. Since the coefficient determinants of the two systems of equations (A1) and (A2) are not equal to zero by Lemma 1 and Lemma 3, \emph{i.e.},
 \begin{equation}
\nonumber
\begin{split}
\left|
  \begin{array}{ccccc}
    1 &1                 &1                        &\cdots &1\\
    1 &\omega_{2}        &(\omega_{2})^{2}         &\cdots &(\omega_{2})^{m-2}\\
    1 &(\omega_{2})^{2}  &[(\omega_{2})^{2}]^{2}   &\cdots &[(\omega_{2})^{2}]^{m-2}\\
    \vdots   &\vdots     &\vdots                   &\ddots &\vdots\\
    1 &(\omega_{2})^{m-2}&[(\omega_{2})^{m-2}]^{2} &\cdots &[(\omega_{2})^{m-2}]^{m-2}
  \end{array}
\right|
\neq 0,
\end{split}
\end{equation}
\begin{equation}
\nonumber
\begin{split}
\left|
  \begin{array}{ccccc}
    1 &1             &1                     &\cdots &1\\
    1 &\overline{\omega_{2}} &\overline{\omega_{2}}^{2} &\cdots &\overline{\omega_{2}}^{m-2}\\
    1        &\overline{\omega_{2}}^{2}  &(\overline{\omega_{2}}^{2})^{2}    &\cdots &(\overline{\omega_{2}}^{2})^{m-2}\\
    \vdots   &\vdots        &\vdots                &\ddots &\vdots\\
    1        &\overline{\omega_{2}}^{m-2}  &(\overline{\omega_{2}}^{m-2})^{2}    &\cdots &(\overline{\omega_{2}}^{m-2})^{m-2}
  \end{array}
\right|
\neq 0,
\end{split}
\end{equation}
we have the unique solutions of Eqs. (A1) and (A2), respectively, \emph{i.e.},
\begin{eqnarray}
\begin{aligned}
a_{01}^{k}=a_{02}^{k}=\cdots=a_{0(m-1)}^{k}=0
\end{aligned}
\end{eqnarray}
and
\begin{eqnarray}
\begin{aligned}
a_{10}^{k}=a_{20}^{k}=\cdots=a_{(m-1)0}^{k}=0
\end{aligned}
\end{eqnarray}
by Lemma 2. Similarly, because $(M_{k}\otimes I_{n\times n})|\phi_{n+m-1}\rangle$ should be orthogonal to $(M_{k}\otimes I_{n\times n})|\phi_{\eta+n}\rangle$ on Alice's side for $\eta=0$, $1$, $\cdots$, $m-2$, we have
\begin{equation}
\left\{
\begin{aligned}
\sum_{j=0}^{m-2}a_{(m-1)j}^{k}=0\qquad\qquad\,\\
\sum_{j=0}^{m-2}(\omega_{2})^{j}a_{(m-1)j}^{k}=0\qquad\\
\sum_{j=0}^{m-2}[(\omega_{2})^{2}]^{j}a_{(m-1)j}^{k}=0\quad\\
\vdots\,\,\,\,\,\,\,\,\,\,\,\,\,\,\,\,\,\,\,\,\,\,\,\,\,\,\,\,\\
\sum_{j=0}^{m-2}[(\omega_{2})^{m-2}]^{j}a_{(m-1)j}^{k}=0
\end{aligned}
\right.
\end{equation} and
\begin{equation}
\left\{
\begin{aligned}
\sum_{j=0}^{m-2}a_{j(m-1)}^{k}=0\qquad\quad\,\,\\
\sum_{j=0}^{m-2}\overline{\omega_{2}}^{j}a_{j(m-1)}^{k}=0\qquad\\
\sum_{j=0}^{m-2}(\overline{\omega_{2}}^{2})^{j}a_{j(m-1)}^{k}=0\quad\\
\vdots\,\,\,\,\,\,\,\,\,\,\,\,\,\,\,\,\,\,\,\,\,\,\,\,\,\,\,\,\,\,\\
\sum_{j=0}^{m-2}(\overline{\omega_{2}}^{m-2})^{j}a_{j(m-1)}^{k}=0\\
\end{aligned}
\right..
\end{equation}
By Lemma 1-3, we get the unique solutions of Eqs. (A5) and (A6), respectively, \emph{i.e.},
\begin{eqnarray}
\begin{aligned}
a_{(m-1)0}^{k}=a_{(m-1)1}^{k}=\cdots=a_{(m-1)(m-2)}^{k}=0,
\end{aligned}
\end{eqnarray}
\begin{eqnarray}
\begin{aligned}
a_{0(m-1)}^{k}=a_{1(m-1)}^{k}=\cdots=a_{(m-2)(m-1)}^{k}=0.
\end{aligned}
\end{eqnarray}

Because the states in $\{(M_{k}\otimes I_{n\times n})|\phi_{\eta+2n+m-2}\rangle:$ $\eta=0,$ $1,$ $\cdots$, $m-2\}$ should be mutually orthogonal on Alice's side, we get the following $m-2$ systems of equations

\begin{equation}
\nonumber
\left\{
\begin{aligned}
\sum_{p=0}^{m-3}[(\omega_{2})^{p}\sum_{j=1}^{m-1}a_{j(p+1)}^{k}]
=-(\omega_{2})^{m-2}\sum_{j=1}^{m-1}a_{j(m-1)}^{k}\qquad\qquad\,\,\\
\sum_{p=0}^{m-3}\{[(\omega_{2})^{2}]^{p}\sum_{j=1}^{m-1}a_{j(p+1)}^{k}\}
=-[(\omega_{2})^{2}]^{m-2}\sum_{j=1}^{m-1}a_{j(m-1)}^{k}\qquad\\
\vdots\qquad \qquad \qquad \qquad \qquad \qquad\quad\,\,\,\\
\sum_{p=0}^{m-3}\{[(\omega_{2})^{m-2}]^{p}\sum_{j=1}^{m-1}a_{j(p+1)}^{k}\}
=-[(\omega_{2})^{m-2}]^{m-2}\sum_{j=1}^{m-1}a_{j(m-1)}^{k}
\end{aligned}
\right.,
\end{equation}
\begin{widetext}
\begin{equation}
\nonumber
\left\{
\begin{aligned}
\sum_{p=0}^{m-3}(\sum_{j=1}^{m-1}\overline{\omega_{2}}^{j-1}a_{j(p+1)}^{k})
=-\sum_{j=1}^{m-1}\overline{\omega_{2}}^{j-1}a_{j(m-1)}^{k}\qquad\qquad\qquad\qquad\qquad\,\,\\
\sum_{p=0}^{m-3}\{[(\omega_{2})^{2}]^{p}\sum_{j=1}^{m-1}\overline{\omega_{2}}^{j-1}a_{j(p+1)}^{k}\}
=-[(\omega_{2})^{2}]^{m-2}\sum_{j=1}^{m-1}\overline{\omega_{2}}^{j-1}a_{j(m-1)}^{k}\qquad\\
\sum_{p=0}^{m-3}\{[(\omega_{2})^{3}]^{p}\sum_{j=1}^{m-1}\overline{\omega_{2}}^{j-1}a_{j(p+1)}^{k}\}
=-[(\omega_{2})^{3}]^{m-2}\sum_{j=1}^{m-1}\overline{\omega_{2}}^{j-1}a_{j(m-1)}^{k}\qquad\\
\vdots \qquad \qquad \qquad \qquad \qquad \qquad \qquad \qquad\\
\sum_{p=0}^{m-3}\{[(\omega_{2})^{m-2}]^{p}\sum_{j=1}^{m-1}\overline{\omega}^{j-1}a_{j(p+1)}^{k}\}
=-[(\omega_{2})^{m-2}]^{m-2}\sum_{j=1}^{m-1}\overline{\omega_{2}}^{j-1}a_{j(m-1)}^{k}
\end{aligned}
\right.,
\end{equation}

\begin{equation}
\nonumber
\left\{
\begin{aligned}
\sum_{p=0}^{m-3}[\sum_{j=1}^{m-1}(\overline{\omega_{2}}^{2})^{j-1}a_{j(p+1)}^{k}]
=-\sum_{j=1}^{m-1}(\overline{\omega_{2}}^{2})^{j-1}a_{j(m-1)}^{k}\qquad\qquad\qquad\qquad\qquad\,\,\,\,\,\\
\sum_{p=0}^{m-3}[(\omega_{2})^{p}\sum_{j=1}^{m-1}(\overline{\omega_{2}}^{2})^{j-1}a_{j(p+1)}^{k}]
=-(\omega_{2})^{m-2}\sum_{j=1}^{m-1}(\overline{\omega_{2}}^{2})^{j-1}a_{j(m-1)}^{k}\qquad\qquad\,\,\,\\
\sum_{p=0}^{m-3}\{[(\omega_{2})^{3}]^{p}\sum_{j=1}^{m-1}(\overline{\omega_{2}}^{2})^{j-1}a_{j(p+1)}^{k}\}
=-[(\omega_{2})^{3}]^{m-2}\sum_{j=1}^{m-1}(\overline{\omega_{2}}^{2})^{j-1}a_{j(m-1)}^{k}\qquad\,\,\\
\vdots \qquad \qquad \qquad \qquad \qquad \qquad \qquad \qquad\quad\\
\sum_{p=0}^{m-3}\{[(\omega_{2})^{m-2}]^{p}\sum_{j=1}^{m-1}(\overline{\omega_{2}}^{2})^{j-1}a_{j(p+1)}^{k}\}
=-[(\omega_{2})^{m-2}]^{m-2}\sum_{j=1}^{m-1}(\overline{\omega_{2}}^{2})^{j-1}a_{j(m-1)}^{k}\\
\end{aligned}
\right.,
\end{equation}
\begin{equation}
\nonumber
\vdots\quad
\end{equation}
\begin{equation}
\nonumber
\left\{
\begin{aligned}
\sum_{p=0}^{m-3}[\sum_{j=1}^{m-1}(\overline{\omega_{2}}^{m-2})^{j-1}a_{j(p+1)}^{k}]
=-\sum_{j=1}^{m-1}(\overline{\omega_{2}}^{d-2})^{j-1}a_{j(m-1)}^{k}\qquad\qquad
\qquad\qquad\qquad\,\,\,\,\,\,\\
\sum_{p=0}^{m-3}[(\omega_{2})^{p}\sum_{j=1}^{m-1}(\overline{\omega_{2}}^{m-2})^{j-1}a_{j(p+1)}^{k}]
=-(\omega_{2})^{m-2}\sum_{j=1}^{m-1}(\overline{\omega_{2}}^{m-2})^{j-1}a_{j(m-1)}^{k}\qquad\qquad\,\,\,\\
\sum_{p=0}^{m-3}\{[(\omega_{2})^{2}]^{p}\sum_{j=1}^{m-1}(\overline{\omega_{2}}^{m-2})^{j-1}a_{j(p+1)}^{k}\}
=-[(\omega_{2})^{2}]^{m-2}\sum_{j=1}^{m-1}(\overline{\omega_{2}}^{m-2})^{j-1}a_{j(m-1)}^{k}\qquad\,\,\\
\vdots \qquad \qquad \qquad \qquad \qquad \qquad \qquad \qquad\qquad\\
\sum_{p=0}^{m-3}\{[(\omega_{2})^{m-3}]^{p}\sum_{j=1}^{m-1}(\overline{\omega_{2}}^{m-2})^{j-1}a_{j(p+1)}^{k}\}
=-[(\omega_{2})^{d-3}]^{m-2}\sum_{j=1}^{m-1}(\overline{\omega_{2}}^{m-2})^{j-1}a_{j(m-1)}^{k}\\
\end{aligned}
\right.,
\end{equation}
\end{widetext}

By Lemma 1-3 and Eqs. (A8), we can get the unique solutions of the above $m-2$ systems of equations, $i.e.,$
\begin{equation}
\left\{
\begin{aligned}
\sum_{j=1}^{m-2}a_{jq}^{k}=a_{(m-1)(m-1)}^{k}\qquad\qquad\qquad\qquad\,\,\,\,\\
\sum_{j=1}^{m-2}\overline{\omega_{2}}^{j-1}a_{jq}^{k}=\overline{\omega_{2}}^{q-1}a_{(m-1)(m-1)}^{k}
\qquad\qquad\,\,\\
\sum_{j=1}^{m-2}(\overline{\omega_{2}}^{2})^{j-1}a_{jq}^{k}=(\overline{\omega_{2}}^{2})^{q-1}
a_{(m-1)(m-1)}^{k}\qquad\\
\vdots \qquad \qquad \qquad \\
\sum_{j=1}^{m-2}(\overline{\omega_{2}}^{m-2})^{j-1}a_{jq}^{k}=(\overline{\omega_{2}}^{m-2})^{q-1}
a_{(m-1)(m-1)}^{k}
\end{aligned}
\right.,
\end{equation}
where $q=1,\,2,\,\cdots,\,m-2.$

By Lemma 1-3 and Eqs. (A9), we have
\begin{equation}
\left\{
\begin{aligned}
a_{jq}^{k}=0\qquad\qquad\,\,\,\\
a_{qq}^{k}=a_{(m-1)(m-1)}^{k}\\
\end{aligned}
\right.
\end{equation}
for $q=1,2,\cdots,m-2;$ $j=1,2,\cdots,m-2$ and $j\neq q.$

Similary, since the states in the set $\{(M_{k}\otimes I_{n\times n})|\phi_{\eta+n}\rangle:$ $\eta=0,$ $1,$ $\cdots$, $m-2\}$ should be orthogonal on Alice's side, we have
\begin{equation}
\left\{
\begin{aligned}
a_{jq}^{k}=0\,\,\,\,\,\,\,\\
a_{qq}^{k}=a_{00}^{k}\\
\end{aligned}
\right.
\end{equation}
for $q=1,\,2,\,\cdots,\,m-2;$ $j=1,\,2,\,\cdots,\,m-2$ and $j\neq q.$

By Eqs. (A3), (A4), (A7), (A8), (A10) and (A11), we have\\
$M_{k}^{\dag}M_{k}=\qquad\qquad\qquad\qquad\qquad\qquad\quad\\$
\begin{equation}
\begin{split}
\left[
  \begin{array}{cccc}
    a_{(m-1)(m-1)}^{k}   &0          &\cdots   &0\\
    0            &a_{(m-1)(m-1)}^{k} &\cdots   &0\\
    0            &0          &\ddots   &0\\
    0            &0          &\cdots   &a_{(m-1)(m-1)}^{k}\\
  \end{array}
\right]
\end{split}\nonumber
\end{equation}
for $k=1,2,\cdots,l.$ This means that all the POVM elements are proportional to  identity matrix. That is, Alice cannot start with a nontrivial measurement to keep the post-measurement states orthogonal.

In fact, Bob will face a similar case as Alice does since the set of these $2(m+n)-4$ states has a symmetrical structure. Therefore, these $2(m+n)-4$ states cannot be exactly distinguished by using only LOCC. This completes the proof.
\end{proof}
\section{Prood of Theorem 5}
\begin{proof}
We prove whether Alice or Bob cannot perform a nontrivial POVM measurement to preserve the orthogonality of the OPSs that are orthogonal only on one side. Without loss of generality, Suppose that Alice firstly performs a general POVM measurement with a set of POVM elements
\begin{equation}
\nonumber
\begin{split}
M_{k}^{\dag}M_{k}=
\left[
  \begin{array}{cccc}
    a_{00}^{k}      &a_{01}^{k}       &\cdots     &a_{0(m-1)}^{k}\\
    a_{10}^{k}      &a_{11}^{k}       &\cdots     &a_{1(m-1)}^{k}\\
    \vdots          &\vdots           &\ddots     &\vdots\\
    a_{(m-1)0}^{k}  &a_{(m-1)1}^{k}   &\cdots     &a_{(m-1)(m-1)}^{k}\\
  \end{array}
\right]
\end{split}
\end{equation}
in the basis $\{|0\rangle,\,|1\rangle,\,|2\rangle,\,\cdots,\, |(m-1)\rangle\}$.

Because $(M_{k}\otimes I_{n\times n})|\phi_{1}\rangle$ is orthogonal to $(M_{k}\otimes I_{n\times n})|\phi_{2n+m-2}\rangle$, $(M_{k}\otimes I_{n\times n})|\phi_{2n+m-1}\rangle$ and $(M_{k}\otimes I_{n\times n})|\phi_{2n+m}\rangle$, \emph{i.e.},
\begin{equation}
\nonumber
\left\{
\begin{aligned}
\langle\phi_{1}|M^{\dag}_{k}M_{k}\otimes I_{n\times n}|\phi_{2n+m-2}\rangle=0\\
\langle\phi_{1}|M^{\dag}_{k}M_{k}\otimes I_{n\times n}|\phi_{2n+m-1}\rangle=0\\
\langle\phi_{1}|M^{\dag}_{k}M_{k}\otimes I_{n\times n}|\phi_{2n+m}\rangle=0\quad\\
\end{aligned}
\right.
\end{equation}
and
\begin{equation}
\nonumber
\left\{
\begin{aligned}
\langle\phi_{2n+m-2}|M^{\dag}_{k}M_{k}\otimes I_{n\times n}|\phi_{1}\rangle=0\\
\langle\phi_{2n+m-1}|M^{\dag}_{k}M_{k}\otimes I_{n\times n}|\phi_{1}\rangle=0\\
\langle\phi_{2n+m}|M^{\dag}_{k}M_{k}\otimes I_{n\times n}|\phi_{1}\rangle=0\quad\\
\end{aligned}
\right.,
\end{equation}
we have
\begin{equation}
\left\{
\begin{aligned}
a_{01}^{k}+a_{02}^{k}+a_{03}^{k}=0\qquad\quad\,\,\\
a_{01}^{k}+\omega a_{02}^{k}+\omega^{2}a_{03}^{k}=0\quad\,\,\,\,\\
a_{01}^{k}+\omega^{2} a_{02}^{k}+(\omega^{2})^{2}a_{03}^{k}=0\\
\end{aligned}
\right.
\end{equation}
and
\begin{equation}
\left\{
\begin{aligned}
a_{10}^{k}+a_{20}^{k}+a_{30}^{k}=0\qquad\quad\,\,\\
a_{10}^{k}+\overline{\omega} a_{20}^{k}+\overline{\omega}^{2}a_{30}^{k}=0\quad\,\,\,\,\\
a_{10}^{k}+\overline{\omega}^{2}a_{20}^{k}+(\overline{\omega}^{2})^{2}a_{30}^{k}=0\\
\end{aligned}
\right.,
\end{equation}
where $\overline{\omega}$ is the conjugate complex number of $\omega$.
By Lemma 1-3, Eqs. (B1) and (B2), we have
\begin{equation}
a_{01}^{k}=a_{02}^{k}=a_{03}^{k}=a_{10}^{k}=a_{20}^{k}=a_{30}^{k}=0.
\end{equation}
Because $(M_{k}\otimes I_{n\times n})|\phi_{1}\rangle$ is orthogonal to $(M_{k}\otimes I_{n\times n})|\phi_{2n+m+2\sigma+1}\rangle$ and
$(M_{k}\otimes I_{n\times n})|\phi_{2n+m+2\sigma+2}\rangle$ for $\sigma=0,\,\, 1,\,\, 2,\,\, \cdots,\,\, d_{1}-3$, we have
\begin{equation}
\left\{
\begin{aligned}
a_{0,2\sigma+4}^{k}+a_{0,2\sigma+5}^{k}=0\\
a_{0,2\sigma+4}^{k}-a_{0,2\sigma+5}^{k}=0\\
\end{aligned}
\right.
\end{equation}
and
\begin{equation}
\left\{
\begin{aligned}
a_{2\sigma+4,0}^{k}+a_{2\sigma+5,0}^{k}=0\\
a_{2\sigma+4,0}^{k}-a_{2\sigma+5,0}^{k}=0\\
\end{aligned}
\right..
\end{equation}
By Eqs. (B4) and (B5), we have
\begin{equation}
\begin{aligned}
a_{0,2\sigma+4}^{k}=a_{0,2\sigma+5}^{k}=a_{2\sigma+4,0}^{k}=a_{2\sigma+5,0}^{k}=0\\
\end{aligned}
\end{equation}
for $\sigma=0,\,\, 1,\,\, 2,\,\, \cdots,\,\, d_{1}-3$.

Because any two states of $\{(M_{k}\otimes I_{n\times n})|\phi_{2n+m-2}\rangle$, $(M_{k}\otimes I_{n\times n})|\phi_{2n+m-1}\rangle$ and $(M_{k}\otimes I_{n\times n})|\phi_{2n+m}\rangle\}$, \emph{i.e.},
\begin{equation}
\nonumber
\left\{
\begin{aligned}
\langle\phi_{2n+m-2}|M_{k}^{\dag}M_{k}\otimes I_{n\times n}|\phi_{2n+m-1}\rangle=0\\
\langle\phi_{2n+m-2}|M_{k}^{\dag}M_{k}\otimes I_{n\times n}|\phi_{2n+m}\rangle=0\quad\\
\end{aligned}
\right.,
\end{equation}
\begin{equation}
\nonumber
\left\{
\begin{aligned}
\langle\phi_{2n+m-1}|M_{k}^{\dag}M_{k}\otimes I_{n\times n}|\phi_{2n+m-2}\rangle=0\\
\langle\phi_{2n+m-1}|M_{k}^{\dag}M_{k}\otimes I_{n\times n}|\phi_{2n+m}\rangle=0\quad\\
\end{aligned}
\right.,
\end{equation}
\begin{equation}
\nonumber
\left\{
\begin{aligned}
\langle\phi_{2n+m}|M_{k}^{\dag}M_{k}\otimes I_{n\times n}|\phi_{2n+m-2}\rangle=0\\
\langle\phi_{2n+m}|M_{k}^{\dag}M_{k}\otimes I_{n\times n}|\phi_{2n+m-1}\rangle=0\\
\end{aligned}
\right.,
\end{equation}
we have
\begin{equation}
\left\{
\begin{aligned}
\sum_{j=1}^{3}a_{j1}+\omega\sum_{j=1}^{3}a_{j2}=-\omega^{2}\sum_{j=1}^{3}a_{j3}\quad\,\,\,\\
\sum_{j=1}^{3}a_{j1}+\omega^{2}\sum_{j=1}^{3}a_{j2}=-(\omega^{2})^{2}\sum_{j=1}^{3}a_{j3}\\
\end{aligned}
\right.,
\end{equation}
\begin{equation}
\left\{
\begin{aligned}
\sum_{i=1}^{2}\sum_{j=1}^{3}\overline{\omega}^{j-1}a_{ji}
=-\sum_{j=1}^{3}\overline{\omega}^{j-1}a_{j3}\qquad\qquad\,\,\\
\sum_{i=1}^{2}[(\omega^{2})^{i-1}\sum_{j=1}^{3}\overline{\omega}^{j-1}a_{ji}]
=-\omega^{4}\sum_{j=1}^{3}\overline{\omega}^{j-1}a_{j3}\\
\end{aligned}
\right.,
\end{equation}
\begin{equation}
\left\{
\begin{aligned}
\sum_{i=1}^{2}\sum_{j=1}^{3}(\overline{\omega}^{2})^{j-1}a_{ji}
=-\sum_{j=1}^{3}(\overline{\omega}^{2})^{j-1}a_{j3}\qquad\quad\\
\sum_{i=1}^{2}[\omega^{i-1}\sum_{j=1}^{3}(\overline{\omega}^{2})^{j-1}a_{ji}]
=-\omega^{2}\sum_{j=1}^{3}(\overline{\omega}^{2})^{j-1}a_{j3}\\
\end{aligned}
\right..
\end{equation}
By Eqs. (B7), (B8) and (B9), we have
\begin{equation}
\left\{
\begin{aligned}
\sum_{j=1}^{3}a_{j1}=\sum_{j=1}^{3}a_{j3}\\
\sum_{j=1}^{3}a_{j2}=\sum_{j=1}^{3}a_{j3}\\
\end{aligned}
\right.,
\end{equation}
\begin{equation}
\left\{
\begin{aligned}
\sum_{j=1}^{3}\overline{\omega}^{j-1}a_{j1}=
\omega^{2}\sum_{j=1}^{3}\overline{\omega}^{j-1}a_{j3}\\
\sum_{j=1}^{3}\overline{\omega}^{j-1}
a_{j2}=\omega\sum_{j=1}^{3}\overline{\omega}^{j-1}a_{j3}\\
\end{aligned}
\right.,
\end{equation}
and
\begin{equation}
\left\{
\begin{aligned}
\sum_{j=1}^{3}(\overline{\omega}^{2})^{j-1}a_{j1}
=\omega\sum_{j=1}^{3}(\overline{\omega}^{2})^{j-1}a_{j3}\\
\sum_{j=1}^{3}(\overline{\omega}^{2})^{j-1}a_{j2}
=\omega^{2}\sum_{j=1}^{3}(\overline{\omega}^{2})^{j-1}a_{j3}\\
\end{aligned}
\right.,
\end{equation}
respectively. By Eqs. (B10), (B11) and (B12), we have
\begin{equation}
\left\{
\begin{aligned}
a_{11}^{k}=a_{33}^{k}\\
a_{21}^{k}=a_{13}^{k}\\
a_{31}^{k}=a_{23}^{k}\\
\end{aligned}
\right.,
\end{equation}
\begin{equation}
\left\{
\begin{aligned}
a_{22}^{k}=a_{33}^{k}\\
a_{12}^{k}=a_{23}^{k}\\
a_{32}^{k}=a_{13}^{k}\\
\end{aligned}
\right..
\end{equation}
Similarly, because any two states of $\{(M_{k}\otimes I_{n\times n})|\phi_{n}\rangle$, $(M_{k}\otimes I_{n\times n})|\phi_{n+1}\rangle$ and $(M_{k}\otimes I_{n\times n})|\phi_{n+2}\rangle\}$, we have
\begin{equation}
\left\{
\begin{aligned}
a_{00}^{k}=a_{22}^{k}\\
a_{10}^{k}=a_{02}^{k}\\
a_{20}^{k}=a_{12}^{k}\\
\end{aligned}
\right.,
\end{equation}
\begin{equation}
\left\{
\begin{aligned}
a_{11}^{k}=a_{22}^{k}\\
a_{01}^{k}=a_{12}^{k}\\
a_{21}^{k}=a_{02}^{k}\\
\end{aligned}
\right..
\end{equation}
By Eqs. (B3), (B13), (B14), (B15) and (B16), we have
\begin{equation}
\left\{
\begin{aligned}
a_{00}^{k}=a_{11}^{k}=a_{22}^{k}=a_{33}^{k}\qquad\qquad\qquad\quad\,\\
a_{12}^{k}=a_{13}^{k}=a_{21}^{k}=a_{23}^{k}=a_{31}^{k}=a_{32}^{k}=0\\
\end{aligned}
\right..
\end{equation}
Because  $(M_{k}\otimes I_{n\times n})|\phi_{2n+m+2\sigma+1}\rangle$ and $(M_{k}\otimes I_{n\times n})$ $|\phi_{2n+m+2\sigma+2}\rangle$ are mutually orthogonal only on Alice's side, \emph{i.e.},
\begin{equation}
\left\{
\begin{aligned}
\langle\phi_{2n+m+2\sigma+1}|M_{k}^{\dag}M_{k}\otimes I_{n\times n}|\phi_{2n+m+2\sigma+2}\rangle=0\\
\langle\phi_{2n+m+2\sigma+2}|M_{k}^{\dag}M_{k}\otimes I_{n\times n}|\phi_{2n+m+2\sigma+1}\rangle=0\\
\end{aligned}
\right.,
\end{equation}
we have
\begin{equation}
\left\{
\begin{aligned}
a^{k}_{2\sigma+4,\,2\sigma+4}=a^{k}_{2\sigma+5,\,2\sigma+5}\\
a^{k}_{2\sigma+4,\,2\sigma+5}=a^{k}_{2\sigma+5,\,2\sigma+4}\\
\end{aligned}
\right.,
\end{equation}
for $\sigma=0,\,\, 1,\,\, 2,\,\, \cdots,\,\, d_{1}-3$. Similarly, because $(M_{k}\otimes I_{n\times n})|\phi_{n+2\sigma+3}\rangle$ and $(M_{k}\otimes I_{n\times n})|\phi_{n+2\sigma+4}\rangle$ are mutually orthogonal only on Alice's side, \emph{i.e.},
we have
\begin{equation}
\left\{
\begin{aligned}
a^{k}_{2\sigma+3,\,2\sigma+3}=a^{k}_{2\sigma+4,\,2\sigma+4}\\
a^{k}_{2\sigma+3,\,2\sigma+4}=a^{k}_{2\sigma+4,\,2\sigma+3}\\
\end{aligned}
\right.,
\end{equation}
for $\sigma=0,\,\, 1,\,\, 2,\,\, \cdots,\,\, d_{1}-3$.

Because any of $\{(M_{k}\otimes I_{n\times n})|\phi_{2n+m-2}\rangle$, $(M_{k}\otimes I_{n\times n})|\phi_{2n+m-1}\rangle$, $(M_{k}\otimes I_{n\times n})|\phi_{2n+m}\rangle\}$ is orthogonal to each of  $\{(M_{k}\otimes I_{n\times n})|\phi_{2n+m+2\sigma+1}\rangle$, $(M_{k}\otimes I_{n\times n})|\phi_{2n+m+2\sigma+2}\rangle\}$ only on Alice's side, we have
\begin{equation}
\begin{aligned}
\langle\phi_{t}|M^{\dag}_{k}M_{k}\otimes I_{n\times n}|\phi_{j}\rangle=0,\\
\end{aligned}
\end{equation}
\begin{equation}
\begin{aligned}
\langle\phi_{j}|M^{\dag}_{k}M_{k}\otimes I_{n\times n}|\phi_{t}\rangle=0,\\
\end{aligned}
\end{equation}
for $t=2n+m-2,\,2n+m-1,\,2n+m$ and $j=2n+m+2\sigma+1,\,2n+m+2\sigma+2.$
By Eqs. (B21) and Eqs. (B22), we have
\begin{equation}
\left\{
\begin{aligned}
a^{k}_{1,\,2\sigma+4}=a^{k}_{2,\,2\sigma+4}=a^{k}_{3,\,2\sigma+4}=0\\
a^{k}_{1,\,2\sigma+5}=a^{k}_{2,\,2\sigma+5}=a^{k}_{3,\,2\sigma+5}=0\\
\end{aligned}
\right.
\end{equation}
and
\begin{equation}
\left\{
\begin{aligned}
a^{k}_{2\sigma+4,\,1}=a^{k}_{2\sigma+4,\,2}=a^{k}_{2\sigma+4,\,3}=0\\
a^{k}_{2\sigma+5,\,1}=a^{k}_{2\sigma+5,\,2}=a^{k}_{2\sigma+5,\,3}=0\\
\end{aligned}
\right.
\end{equation}
respectively, for $\sigma=0,\,\, 1,\,\, 2,\,\, \cdots,\,\, d_{1}-3$.

Because $(M_{k}\otimes I_{n\times n})|\phi_{2n+m+2\sigma+1}\rangle$ and $(M_{k}\otimes$ $ I_{n\times n})|\phi_{2n+m+2\sigma+2}\rangle$ are orthogonal to $(M_{k}\otimes$ $I_{n\times n})|\phi_{2n+m+2\lambda+1}\rangle$ and $(M_{k}\otimes I_{n\times n})|\phi_{2n+m+2\lambda+2}\rangle$ only on Alice's side for $\sigma,$ $\lambda$ = $0$, $1$, $2$, $\cdots$, $d_{1}-3$ and $\sigma\neq \lambda$, \emph{i.e.},
\begin{equation}
\nonumber
\left\{
\begin{aligned}
\langle\phi_{2n+m+2\sigma+1}|M_{k}^{\dag}M_{k}\otimes I_{n\times n}|\phi_{2n+m+2\lambda+1}\rangle=0\\
\langle\phi_{2n+m+2\sigma+1}|M_{k}^{\dag}M_{k}\otimes I_{n\times n}|\phi_{2n+m+2\lambda+2}\rangle=0\\
\langle\phi_{2n+m+2\sigma+2}|M_{k}^{\dag}M_{k}\otimes I_{n\times n}|\phi_{2n+m+2\lambda+1}\rangle=0\\
\langle\phi_{2n+m+2\sigma+2}|M_{k}^{\dag}M_{k}\otimes I_{n\times n}|\phi_{2n+m+2\lambda+2}\rangle=0\\
\end{aligned}
\right.,
\end{equation}
so we have
\begin{equation}
\nonumber
\left\{
\begin{aligned}
\sum_{i=2\lambda+4}^{2\lambda+5}a_{2\sigma+4,i}^{k}+
\sum_{i=2\lambda+4}^{2\lambda+5}a_{2\sigma+5,i}^{k}=0\qquad\qquad\quad\\
\sum_{i=2\lambda+4}^{2\lambda+5}(-1)^{i}a_{2\sigma+4,i}^{k}+
\sum_{i=2\lambda+4}^{2\lambda+5}(-1)^{i}a_{2\sigma+5,i}^{k}
=0\quad\\
\sum_{i=2\lambda+4}^{2\lambda+5}a_{2\sigma+4,i}^{k}-
\sum_{i=2\lambda+4}^{2\lambda+5}a_{2\sigma+5,i}^{k}
=0\qquad\qquad\quad\\
\sum_{i=2\lambda+4}^{2\lambda+5}(-1)^{i}a_{2\sigma+4,i}^{k}+
\sum_{i=2\lambda+4}^{2\lambda+5}(-1)^{i+1}a_{2\sigma+5,i}^{k}
=0\\
\end{aligned}
\right..
\end{equation}
Thus we get
\begin{equation}
\begin{aligned}
a_{2\sigma+4,2\lambda+4}^{k}=a_{2\sigma+4,2\lambda+5}^{k}=0,\\
a_{2\sigma+5,2\lambda+4}^{k}=a_{2\sigma+5,2\lambda+5}^{k}=0,
\end{aligned}
\end{equation}
where $\sigma,$ $\lambda$ = $0$, $1$, $2$, $\cdots$, $d_{1}-3$ and $\sigma\neq \lambda$. Similarly, because $(M_{k}\otimes I_{n\times n})|\phi_{n+2\sigma+3}\rangle$ and $(M_{k}\otimes I_{n\times n})|\phi_{n+2\sigma+4}\rangle$ are orthogonal to $(M_{k}\otimes I_{n\times n})|\phi_{n+2\lambda+3}\rangle$ and $(M_{k}\otimes I_{n\times n})|\phi_{n+2\lambda+4}\rangle$ for $\sigma,$ $\lambda$ = $0$, $1$, $2$, $\cdots$, $d_{1}-3$ and $\sigma\neq \lambda$, we have
\begin{equation}
\begin{aligned}
a_{2\sigma+3,2\lambda+3}^{k}=a_{2\sigma+3,2\lambda+4}^{k}=0,\\
a_{2\sigma+4,2\lambda+3}^{k}=a_{2\sigma+4,2\lambda+4}^{k}=0,
\end{aligned}
\end{equation}
where $\sigma,$ $\lambda$ = $0$, $1$, $2$, $\cdots$, $d_{1}-3$ and $\sigma\neq \lambda$.

By Eqs. (B3), (B6), (B17), (B19), (B20), (B23), (B24), (B25) and (B26), we get
\begin{equation}
\nonumber
\begin{split}
M_{k}^{\dag}M_{k}=
\left[
  \begin{array}{cccc}
    a_{00}^{k}  &0           &\cdots     &0\\
    0           &a_{00}^{k}  &\cdots     &0\\
    \vdots      &\vdots      &\ddots     &\vdots\\
    0           &0           &\cdots     &a_{00}^{k}\\
  \end{array}
\right]_{m\times m.}
\end{split}
\end{equation}
This means that Alice only can perform a trivial measurement to preserve the orthogonality of the post-measurement states. So does Bob because of the symmetry of the set of these states. Therefore, the set is nonlocal, \emph{i.e.}, these states cannot be perfectly discriminated by LOCC. This completes the proof.
\end{proof}
\vbox{}

\nocite{*}
\bibliography{apssamp}
\end{document}